\documentclass[12pt,a4paper]{article}
\usepackage{jheppub}

\usepackage{color}
\usepackage{amsmath}

\usepackage{verbatim}   
\usepackage{subfigure}  
\usepackage{amsfonts}
\usepackage{amssymb}
\usepackage{mathrsfs}
\usepackage{graphicx}
\usepackage{slashed}
\usepackage{amsthm}

\newcommand{\half}{\frac{1}{2}}

\newcommand{\IC}{\mathbb{C}}
\newcommand{\IR}{\mathbb{R}}

\def\bar#1{\overline{#1}}

\def\inv{^{\raise.15ex\hbox{${\scriptscriptstyle -}$}\kern-.05em 1}}
\def\lbar{{\lower.35ex\hbox{$\mathchar'26$}\mkern-10mu\lambda}} 

\def\to{\rightarrow}

\let\p=\partial
\let\<=\langle
\let\>=\rangle

\let\+=\uparrow

\def\a{\alpha'}

\def\d{\textrm{d}}
\def\End{\textrm{End}}
\def\Im{\textrm{Im}}
\def\ker{\textrm{ker}}

\def\bp{\bar{\partial}}
\def\tr{\textrm{tr}}
\def\OO{\mathcal{O}}

\def\Q{\mathcal{Q}}

\def\H{\mathcal{H}}

\def\B{\mathcal{B}}
\def\Z{\mathcal{Z}}
\def\F{\mathcal{F}}

\def\H{\mathcal{H}}

\newtheorem{Cor}{Corollary}
\newtheorem{Theorem}{Theorem}

\newtheorem{Proposition}{Proposition}

\theoremstyle{definition}

\begin{document}



\title{\rm Holomorphic Bundles and the Moduli Space of\\
N=1 Supersymmetric Heterotic Compactifications\\}

\author[a]{Xenia de la Ossa}
\author[b]{and Eirik E.~Svanes}

\affiliation[a]{Mathematical Institute\\
University of Oxford\\
Andrew Wiles Building\\ Radcliffe Observatory Quarter\\
Woodstock Road\\
Oxford OX2 6GG, UK\\}

\affiliation[b]{Rudolf Peierls Centre for Theoretical Physics\\ 
University of Oxford\\ 
1 Keble Road\\
Oxford OX1 3NP, UK\\}

\emailAdd{delaossa@maths.ox.ac.uk}
\emailAdd{e.svanes1@physics.ox.ac.uk}

\null\vskip10pt

\abstract{ We describe the first order moduli space of heterotic string theory compactifications which preserve $N=1$ supersymmetry in four dimensions, that is, the infinitesimal parameter space of the Strominger system. We establish that if we promote a connection on $TX$ to a field, the moduli space corresponds to deformations of a holomorphic structure $\bar D$ on a bundle $\cal Q$. The bundle $\cal Q$ is constructed as an extension by the cotangent bundle $T^*X$ of the bundle $E= {\rm End}(V) \oplus {\rm End}(TX) \oplus TX$ with an extension class $\cal H$ which precisely enforces the anomaly cancelation condition. The deformations corresponding to the bundle $E$ are simultaneous deformations of the holomorphic structures on the poly-stable bundles $V$ and $TX$ together with those of the complex structure of $X$. We discuss the fact that the ``moduli'' corresponding to ${\rm End}(TX)$ cannot be physical, but are however needed in our mathematical structure to be able to enforce the anomaly cancelation condition. In the Appendix we comment on the choice of connection on $TX$ which has caused some confusion in the community before. It has been shown by Ivanov and others that this connection should also satisfy the instanton equations, and we give another proof of this fact.}


\maketitle

\newpage

\section{Introduction.}
\label{sed:intro}
Ever since Strominger and Hull~\citep{Strominger:1986uh, Hull:1986kz} first worked out the geometry arising from a general  supersymmetric compactification of the heterotic string assuming that the four dimensional space time is maximally symmetric and that the compactification preserves $N=1$ supersymmetry\footnote{See also \citep{Lust:1986ix} where heterotic torsional compactifications on coset spaces were studied for the first time.}, torsional compactifications of the heterotic string have been an active topic of interest
both from the ten-dimensional supergravity point of view \citep{Lust:1986ix, Bergshoeff1989439, Bergshoeff:1988nn, Dasgupta:1999ss, 
1999math......1090A, Ivanov:2000fg, 2001CQGra..18.1089I, 2001math......2142F, Becker:2002sx, Gauntlett:2003cy, Becker:2003yv, Becker:2003sh, LopesCardoso:2003sp, LopesCardoso:2003af, Becker:2005nb, Becker:2006xp, Lapan2006, Becker:2006et, Fu:2006vj, Becker:2009df,  Andreas:2010cv, 2012CMaPh.315..153A, 2013arXiv1304.4294G}, 
and from the two-dimensional sigma-model point of view \citep{Adams:2006kb, Sharpe:2008rd, Kreuzer:2010ph, McOrist:2010ae, Beccaria:2010yp, NibbelinkGroot:2010wm, McOrist:2011bn,  Melnikov:2011ez, Blaszczyk:2011ib, Quigley:2011pv, Adams:2012sh, Nibbelink:2012wb, Melnikov:2012hk, Quigley:2012gq, Melnikov:2012nm}.  

In this paper, we study the heterotic string from the ten-dimensional supergravity point of view. We review compactifications corresponding to the Strominger system which appears at first order in the $\a$ expansion of the heterotic theory and study the infinitesimal moduli space.  In the Appendix, we give a brief motivation for the Strominger system and consider higher order $\a$ corrections to the system. We will however delay a full treatment of the higher order theory to a forthcoming  publication \citep{delaossa2014}.

\subsection{Heterotic Supergravity and the Strominger System.}

We begin in section \ref{sec:hetcomp} with a review of  the geometry of the Strominger system and set up the notation. 
The heterotic compactification we are interested in consists on a pair $(X,V)$ where $X$ is a six dimensional Riemannian spin manifold, together with a vector bundle $V$ on $X$.  This pair has the geometric properties  given by
\begin{itemize}
\item A six-dimensional compact space $X$ with an $SU(3)$-structure given by a nowhere vanishing three-form $\Psi$, and a hermitian form $\omega$ satisfying the $SU(3)$-structure compatibility conditions
\begin{equation*}
\Psi\wedge\omega=0\:,\;\;\;\;\;\;\frac{i}{\vert\vert\Psi\vert\vert^2}\Psi\wedge\bar\Psi=\frac{1}{6}\omega\wedge\omega\wedge\omega\:.
\end{equation*}
The space $X$ is complex, with complex structure determined by the form $\Psi$  which is conformally holomorphic,
\begin{equation*}
\bp(e^{-2\phi}\Psi)=0\:,
\end{equation*}
and the hermitian form $\omega$ is conformally balanced,
\begin{equation*}
\d(e^{-2\phi}\omega\wedge\omega)=0\:.
\end{equation*}
In this paper, manifolds with this structure are called manifolds with a {\it heterotic $SU(3)$ structure}. 
\item A gauge-bundle connection $A$ with structure group contained in $E_8{\times}E_8$ on a vector bundle $V$ with field strength $F$ satisfying the instanton condition. That is,  the bundle is holomorphic, and the curvature satisfies the Yang-Mills condition
\begin{equation*}
\omega\lrcorner F=0\:.
\end{equation*}
\item A  connection $\nabla^I$ on the tangent bundle with curvature $R^I$ which also satisfies the instanton condition. As discussed in the Appendix, this instanton connection is needed to ensure that supersymmetric solutions which satisfy the anomaly cancelation condition, also solve the equations of motion\footnote{To $\cal{O}(\alpha')$ this instanton condition is satisfied by the Hull connection.}.  In this work however, we promote this instanton connection to a field.  As we will see, this is needed to be able to implement the anomaly cancelation condition into the deformation theory.  The price we pay is that in doing so we get extra ``moduli" associated to this connection, which will have to be given the correct physical interpretation.  We leave the interpretation of these parameters for a forthcoming publication~\cite{delaossa2014}.
\item The connections and the geometry of $X$ are related by the Bianchi identity
\begin{equation*}
-2i\p\bp\omega=\frac{\a}{4}(\tr (F\wedge F)-\tr (R^I\wedge R^I))\:.
\end{equation*}
\end{itemize}
We would like to remark that due to the theorem of Li and Yau~\citep{MR915839}, we require the holomphic bundles $V$ and $TX$ to be polystable in order for there to be a solution of the instanton equations. We would also like to point out that the geometry coming from requiring a supersymmetric theory to {\it second order in the $\alpha'$ expansion} can also be described in terms of the Strominger system, at least in the case when the base $X$ is compact. We will however defer the discussion of higher order $\alpha'$ corrections to a future companion paper \cite{delaossa2014}, but we comment briefly on the results of that paper in the Appendix.

\subsection{Holomorphic Structures and Moduli.}
We begin subsection \ref{subsec:infsu3} by discussing the deformations of the relevant $SU(3)$ structure of $X$.  Such a deformation of the $SU(3)$ structure corresponds to simultaneous deformations of the complex structure determined by $\Psi$ together with those of the hermitian structure determined by $\omega$, such that the heterotic $SU(3)$ structure is preserved. The deformations of the complex structure are easily described as the complex structure does not depend on the metric or the hermitian structure on $X$.  The analysis is similar to some extent to that for Calabi--Yau manifolds. Deformations of the hermitian structure satisfying the conformally balanced condition is more difficult as $\omega$ also deforms with the complex structure. One of the problems is that the moduli space of deformations of the hermitian structure seems to be infinite dimensional ~\citep{Becker:2006xp}.  Moreover, the conformally balanced condition is not stable under deformations of the complex structure~\citep{MR1107661, MR1137099, 0735.32009, 0793.53068} in contrast with the stability of the K\"ahler condition.  It turns out that by including the equations derived from the deformation of the anomaly cancelation condition, we find a finite dimensional space for these parameters (see subsections \ref{subsubsec:defHS} and \ref{subsec:Anomaly}).  We note that in a forthcoming publication~\citep{DKS2014}, we show that one parameter families of manifolds with a heterotic $SU(3)$ structure which have an integrable $G_2$ structure, or a certain $SU(4)$ or $\text{Spin}(7)$ structure, are families which automatically have the conformally balanced condition preserved along the family.  This could be very interesting for applications to M-theory and F-theory. 

We investigate the moduli of the Strominger system using the mathematical tools available in  deformation theory of holomorphic structures.
We use the machinery developed by Atiyah \cite{MR0086359}.\footnote{This was also  used in \cite{Anderson:2010mh} where the combined bundle and complex structure moduli where studied in the Calabi-Yau case.}  We construct a holomorphic structure $\bar D$ on an extension bundle $\Q$ which
is an extension by the holomorphic cotangent bundle $T^*X$ of a bundle $E$ given by the short exact sequence
\[ 0 \rightarrow T^*X \rightarrow \Q \rightarrow E \rightarrow 0~,\]
with an extension class $\cal H$ which precisely enforces the anomaly cancelation condition. We  compute first order deformations of the holomorphic structure by computing a long exact sequence in cohomology associated to the short exact sequence above.  We proceed in a stepwise manner.

In subsection \ref{subsec:defsV} we study in detail the deformations  of the holomorphic structure of the bundle $V$ on $X$. We generalise the work in~\citep{Anderson:2010mh, Anderson:2011ty} for the case in which $X$ is a Calabi-Yau manifold to the case of the more general non-K\"ahler conformally balanced manifolds at hand in view of the theorem by Li and Yau~\citep{MR915839}. 
Recall for instance that the holomorphic condition on the gauge bundle is equivalent to a nilpotency condition on the operator
\begin{equation*}
\bp_{\cal A}=\bp+[{\cal A},\:],
\end{equation*}
that is $\bp_{\cal A}^2=0$. Here $\cal A$ is the $(0,1)$-part of the gauge connection. Moreover, simultaneous deformations of the complex structures on $(X,V)$ correspond to elements in the cohomology groups
\begin{equation*}
H^{(0,1)}_{\bp_A}(X,\End(V))\oplus\ker(\F)\:,
\end{equation*}
where $TX$ is the holomphic tangent bundle, and
\begin{equation*}
\F\::\:H^{(0,1)}_{\bp}(X,TX)\rightarrow H^{(0,2)}_{\bp_A}(X,\End(V))~,
\end{equation*}
 is the Atiyah map associated to the extension of $TX$ by $\End(V)$. Here $H^{(0,1)}_{\bp_A}(\End(V))$ are the bundle moduli, and they correspond to the gauge fields in the lower four dimensional theory.
 We extend our results in section \ref{subsec:defsTX} to include the deformations of the instanton connection $\nabla^I$ on $TX$ by considering the extension bundle
 \[ E= \End(TX)\oplus \End(V)\oplus TX~.\]
 We then obtain simultaneous deformations of the holomorphic structures on the bundles and of the complex structure of $X$ by computing the deformations of a holomorphic structure on this extension bundle. The moduli associated to deformations of the connection $\nabla^I$ which leave $X$ fixed are given by elements in the cohomology group
\begin{equation*}
H^{(0,1)}_{\bp_I}(X,\End(TX))\:.
\end{equation*}
These are the extra ``moduli'' which appear as a consequence of considering $\nabla^I$ as a field. These fields however cannot be true moduli of the theory since eventually one has to remember that the connection for the curvature $R$ in the 10 dimensional heterotic action depends on the other fields in the theory (precisely what this dependence is has to do with how the fields are defined~\citep{Sen1986289}).  

In section \ref{subsec:Primitive}, we discuss the fact that the primitivity condition for the curvature of the instanton connections on the bundles also has a solution after deformations of the holomorphic structure on $E$. As the stability of the bundles is preserved under deformations of the complex structure of $X$ which preserve the holomorphicity of the bundles~\citep{MR2665168},  the theorem of Li and Yau~\citep{MR915839} now guarantees that the deformed connections satisfy the instanton conditions. Nevertheless, we prove that on a conformally balanced manifold, a general first order variation (including varying the hermitian structure) of the primitivity conditions of the curvatures preserving the primitivity conditions, does not pose  any new constraints on the moduli space. 

 The full moduli of the Strominger system is given in \ref{subsec:Anomaly} where we state the main result of the paper.
As mentioned above, we define $\Q$ by extending $E$ by the holomorphic co-tangent bundle $T^*X$
 \begin{equation*}
\Q=T^*X\oplus\End(V)\oplus\End(TX)\oplus TX\:.
\end{equation*}
and define a holomorphic structure $\bar D$  on $\cal Q$
\begin{equation*}
\bar D:\Omega^{(p,q)}(X,\Q)\rightarrow\Omega^{(p,q+1)}(X,\Q)\:.
\end{equation*}
The operator $\bar D$ has a rather lengthy definition, which we leave to section \ref{subsec:Anomaly}. What is important is that  
 it includes the anomaly cancelation condition.  Moreover $\bar D^2 = 0$ if and only if the Bianchi identities for $F$, $R^I$ and the anomaly cancelation condition are satisfied.  
The first order deformations of this holomorphic structure correspond to the elements in the cohomology group $H^{(0,1)}_{\bar D}(\Q)$, that is, the tangent space $\cal TM$ to the moduli space $\cal M$ is given by 
\begin{equation*}
\mathcal{T M}\cong H^{(0,1)}_{\bar D}(\Q)\cong  \Big[H^{(0,1)}_{\bp}(T^*X)\Big/\textrm{Im}(\H)\Big]\oplus\ker(\H)\:,
\end{equation*}
and we show that this is the infinitesimal moduli space of the heterotic compactifications. The subgroup $\ker(\H)$ is contained in the moduli space of deformations of $E$, that is the simultaneous variations of complex structure on X and holomorphic structures on the bundles, and it is in fact the kernel of a map $\cal H$ that corresponds to the analogue of the Atiyah class  for the  short exact extension sequence defining $\Q$.
This is nothing but the obvious fact that the anomaly cancellation condition poses a non-trivial extra constraint on the moduli.  We also argue that
the (complexified) hermitian parameters belong to the group
\begin{equation*}
H^{(0,1)}_{\bp}(X,T^*X)\:,
\end{equation*}
like in the Calabi--Yau case.  These should be modded out by 
\[\textrm{Im}(\H)\cong\{\tr(F\,\alpha)\:\vert\:\alpha\in H^0(X,\End(V))\}\subset H^{(1,1)}(X),\] 
which appears whenever the bundle $V$ is polystable, and which precisely enforces the Yang-Mills condition on $V$. We remark that in~\citep{Melnikov:2011ez}, the authors study this structure in the limiting case when
$\alpha'\to 0$ for heterotic compactifications from the two dimensional $(0,2)$ superconformal field theory.

We devote the last section \ref{sec:conclusions} to a discussion of our results, in particular the fact that we appear to find in our set up extra moduli, those in $H^{(0,1)}_{\bp_I}(\End(TX))$ which correspond to deformations of the holomorphic structure of the tangent bundle given by $\nabla^I$ but which leave $X$ fixed.  It seems that we need this extra structure to be able to enforce the anomaly cancelation condition and, as discussed in the Appendix, this seems the natural mathematical structure to consider.  Note that this structure, in which the connection $\nabla^I$ behaves as another dynamical field,  is also the natural one appearing when one considers the heterotic theory to 
higher orders in the $\alpha'$ expansion~\citep{delaossa2014}.  Moreover, the mathematical structure we have in this paper is very similar to that in~\citep{2013arXiv1304.4294G} where a generalised geometry for the heterotic supergravity is proposed, and in \citep{2013arXiv1308.5159B}, where $T$-duality for heterotic compactifications in the context of generalised geometry and Cournat algebroids is studied.
Of course, there are however reasons as to why we do not want these extra fields in the low energy field theory, namely, that the connection $\nabla^I$ is not independent of the geometry of $X$.  We will discuss these points in the conclusion section and also in a future publication.

\newpage

\section{SU(3) Structures and the Strominger System.}
\label{sec:hetcomp}

We begin with a review of the results of Hull~\citep{Hull:1986kz}, Strominger~\citep{Strominger:1986uh} and 
de\,Wit\,{\it et\,al}~\citep{deWit:1986xg}.
The requirements that the four dimensional space--time is maximally symmetric, that $N=1$ supersymmetry is preserved in four dimensions, and that  an equation cancelling anomalies is satisfied, pose strong constraints on the possible geometries that are allowed as solutions of the equations of motion.   

Recall first the massless spectrum of the $N=1$ ten dimensional supergravity theory.  The bosonic fields are: the metric $g_{MN}$, the dilaton $\Phi$, a totally antisymmetric 3-form $H_{MNP}$ and the $E_8{\times}E_8$ gauge field strength $F_{MN}$.  The fermionic fields are the spin 3/2 gravitino, the spin 1/2 dilatino, and the spin 1/2 gluinos which take values in the adjoint representation of $E_8{\times}E_8$.   The indices $M, N$, etc, are ten dimensional tangent space indices.

A compactification to four dimensions is obtained by considering a ten dimensional space-time which is a local product $M_4\times X$ of a maximally symmetric four dimensional space-time $M_4$, and a six dimensional manifold $X$.  Maximal symmetry on $M_4$ requires that on the internal manifold $X$ there are: a scalar $\phi$ (the dilaton),  a Riemannian metric $g_{mn}$, a 3-form $H$ (the flux),  gauge fields $A_m$ in the adjoint representation of a subgroup of $E_8{\times}E_8$, and the supersymmetric partners of these bosonic fields.  The latin indices $m,n$, etc, are six dimensional indices on the tangent space $TX$ of $X$.  Also, a consequence of imposing the constraints of  $N=1$ supergravity in four dimensional space time is that $M_4$ must be Minkowski. 

So $X$ is a real six dimensional manifold with metric $g$,  and on $X$ there is a vector bundle $V$ with curvature $F$ which takes values in $\End(V)$.  We now discuss the constraints on the geometry of $(X,V)$. 

\subsection{Constraints on the Geometry of \texorpdfstring{$X$}\ .}
\label{subsec:Xgeom}

Supersymmetry requires that on $X$ there must exist  a nowhere vanishing globally defined complex spinor $\eta$.
This means that $X$ must be a spin manifold and that the structure group of  $X$ is reduced to 
a subgroup $SU(3)\subset \text{Spin}(6)$.  
An $SU(3)$ structure on $X$~\citep{0444.53032, 1024.53018, LopesCardoso:2002hd}  is defined by a triple $(X,\omega,\Psi)$, where $\omega$ is a non-degenerate globally well defined real 2-form, and $\Psi$ is a locally decomposable no-where vanishing globally well defined complex 3-form.  
The forms $\Psi$ and $\omega$ satisfy
\begin{equation}
 \omega\wedge\Psi = 0~.\label{eq:compatibility}
 \end{equation}
In fact, there is an $SU(3)$ structure on $X$ determined entirely by the spinor $\eta$.  The two non-degenerate forms, $\omega$ and $\Psi$, can be constructed as bilinears of  $\eta$
\begin{equation*}
\begin{split}
\omega_{mn} &= -i \, \eta^\dagger \,\gamma_{mn}\, \eta\\
\Psi_{mnp} &= ~~\eta^T\,\gamma_{mnp}\,\eta~,
\end{split}
\end{equation*}
where $\gamma_m$ are the matrices that satisfy the Clifford algebra in six dimensions
\[ \{\gamma_m,\gamma_n\} = 2\, g_{mn}~,\]
and $\gamma_{m_1m_2\cdots m_p}$ denotes the totally antisymmetric product of $p$ gamma matrices
\[\gamma_{m_1m_2\cdots m_p} = \gamma_{[m_1}\gamma_{m_2}\cdots\gamma_{m_p]}~.\] 
Using Fierz rearrangement, one can prove that these satisfy  \eqref{eq:compatibility}.  One can also prove that there is a unique (up to a constant) invariant volume form on $X$ which satisfies
the compatibility condition
\begin{equation}
\d{\rm vol}_X =  \frac{1}{6}\, \omega\wedge\omega\wedge\omega =
\frac{i}{||\Psi||^2}\, \Psi\wedge\bar\Psi~,\label{eq:volume}
\end{equation} 
where
\[ ||\Psi||^2 = \frac{1}{3!}\, \Psi^{mnp}\, \Psi_{mnp}~.\]

The (real part of the) complex 3-form $\Psi$ determines a unique almost complex structure $J$~\citep{MR1871001} such that $\Psi$ is a $(3,0)$-form
with respect to $J$.  In fact, 
\begin{equation}
\label{eq:complexstr}
{J_m}^n=\frac{{I_m}^n}{\sqrt{-\frac{1}{6}\tr I^2}},
\end{equation}
where the tangent bundle endomorphism $I$ is given by
\begin{equation}
{I_m}^n=(\textrm{Re}\Psi)_{mpq}(\textrm{Re}\Psi)_{rst}\,\epsilon^{npqrst}.\label{eq:CSnorm}
\end{equation}
With the normalization in \eqref{eq:complexstr},  it is not too difficult to prove that $J^2=-{\bf 1}$.  Note also that a change of scale $\Psi\to\lambda\Psi~,\lambda\in\IC^*$, defines the same complex structure $J$. 
With respect to $J$, the real two form $\omega$ is type $(1,1)$ due to \eqref{eq:compatibility}.  Moreover, $\omega$ is an almost hermitian form
\begin{equation}
\omega(X,Y) =  g(JX, Y)~,\quad \forall X, Y\in TX~.
\end{equation}

Therefore 6-dimensional manifolds with an $SU(3)$ structure are almost hermitian manifolds with trivial canonical bundle. We are not done however, as the preservation of  supersymmetry  also imposes differential conditions on $\omega$ and $\Psi$. 

Preservation of supersymmetry requires that on $X$ there must exist a metric connection $\nabla^+$ with skew-symmetric torsion $T=H$, where $H$ is the 3-form flux.  In other words, the connection symbols are
\begin{equation}
 \Gamma^+_{mn}{}^p = \Gamma^{LC}_{mn}{}^{\, p} + \frac{1}{2}\, H_{mn}{}^p~,\label{eq:BismutGamma}
 \end{equation}
where $\Gamma^{LC}_{mn}{}^{\, p}$ are the Christoffel symbols.
The vanishing of the supersymmetric variation of the graviton requires that the spinor $\eta$ must be covariantly constant with respect to this connection
\[ \nabla^+_m\,\eta = \nabla^{LC}_m\, \eta + \frac{1}{8}\, H_{mnp}\, \gamma^{np}\, \eta = 0~.\]
This in turn means that the forms $\omega$ and $\Psi$ are covariantly constant
\[ \nabla^+\Psi = 0~,\qquad \nabla^+\omega= 0~.\]
The almost complex structure determined by $\Psi$ must also be covariantly constant
\[ \nabla^+ J = 0~,\]
that is $\nabla^+$ is a hermitian connection. It is straightforward to prove that Nijenhuis tensor of $J$ vanishes. Therefore $J$ is integrable and $X$ {\it must be a complex manifold}.

On a complex manifold there is unique metric hermitian connection which has  totally antisymmetric torsion, and this is precisely the connection $\nabla^+$.  This connection is called the {\it Bismut connection}~\citep{0666.58042} in the mathematics literature, and its torsion is given by
\begin{equation}
T= H = i(\partial - \bar\partial)\,\omega~.\label{eq:Bismut}
\end{equation}

Equations for the exterior derivative of $\omega$ and $\Psi$ can be obtained from the fact that both $\omega$ and $\Psi$ are covariantly constant.  One can decompose the exterior derivative of $\omega$ and $\Psi$ into irreducible representations of $SU(3)$.
For a general $SU(3)$ structure,   we have 
\begin{align*}
\d\omega&=-\frac{12}{\vert\vert\Psi\vert\vert^2}\,\textrm{Im}(W_0\bar\Psi)+W_1^\omega\wedge\omega+W_3\\
\d\Psi&=W_0\ \omega\wedge\omega+W_2\wedge\omega+\bar W_1^\Psi\wedge\Psi\:.
\end{align*}
were $(W_0, W _1^\omega,W_1^\Psi,W_2,W_3)$ are the five {\it torsion classes}~\citep{0444.53032, 1024.53018, LopesCardoso:2002hd, Gauntlett:2003cy}.
Here, $W_0$ is a complex function, $W_2$ is a primitive $(1,1)$-form, $W_3$ is  a real primitive 3-form of type $(1,2)+(2,1)$, $W_1^\omega$ is a real one-form, and $W_1^\Psi$ is a $(1,0)$-form. 
The one forms $W_1^\omega$ and $W_1^\Psi$ are known as the Lee-forms of $\omega$ and $\Psi$ respectively, and they are given by
\begin{align*}
W_1^\omega&=\frac{1}{2}\, \omega\lrcorner\d\omega\\
W_1^\Psi&=- \frac{1}{\vert\vert\Psi\vert\vert^2}\,\Psi\lrcorner\d\bar\Psi.
\end{align*}
The contraction operator $\lrcorner$ is defined as 
\[ \alpha\lrcorner \beta = \frac{1}{k! p!}\, \alpha^{m_1\cdots m_k}\, \beta_{m_1\cdots m_k n_1\cdots n_p}\,
\d x^{n_1}\wedge\cdots\wedge \d x^{n_p}
= (-1)^{p(d-p-k)}\,* (\alpha\wedge *\beta)~,\]
where $\alpha$ is a $k$-form, $\beta$ is a $k+p$-form, and $d$ is the dimension of the manifold, which in our case is $d=6$.

Under a change of scale $\Psi\rightarrow\lambda\Psi$, $\lambda\in\IC^*$, the Lee-forms $W_1^\omega$ and $W_1^{\Psi}$, and $W_3$ are invariant, however
$$W_0\longrightarrow\lambda W_0~,\qquad W_2\longrightarrow\lambda W_2~.$$

The vanishing of the Nijenhuis tensor, and thus the integrability of the complex structure $J$, is equivalent to the vanishing of $W_0$ and $W_2$ (note that these are the only torsion classes that scale under $\Psi\to\lambda\Psi$).
Also, for the Bismut connection the Lee-forms are related by~\citep{LopesCardoso:2002hd}
\[   {\rm Re}(W_1^\Psi) =  W_1^\omega~.\]
Therefore, the exterior derivatives of $\omega$ and $\Psi$ are
\begin{align*}
\d\omega&= {\rm Re}(W_1^\Psi)\wedge\omega+W_3\:\\
\d\Psi&=\bar W_1^\Psi\wedge\Psi\:.
\end{align*}
Note that $\bar\partial\, \bar W_1^\Psi = 0$ as can be seen by taking the exterior derivative of the second equation. 

The vanishing of the supersymmetric variation of the dilaton gives a further constraint: the Lee-form of $\omega$ must be exact with
\[ W_1^\omega = \d\phi~.\]
Therefore, the equation for $\d\omega$ gives
\begin{equation}
\d( e^{-2\phi}\, \omega\wedge\omega) = 0~,\label{eq:confbal}
\end{equation}
that is, the manifold $X$ is required to be {\it conformally balanced}.
Furthermore, the equation for $\d\Psi$ means that $X$ must have a  {\it holomorphically trivial canonical bundle}
\begin{equation}
 \d(  e^{-2\phi}\, \Psi) = 0~.\label{eq:hol}
 \end{equation}
 In this paper, a complex conformally balanced manifold $X$ with a holomorphically trivial canonical bundle will be called a manifold with a {\it heterotic structure}.
 
 \subsection{Constraints on the Vector Bundle \texorpdfstring{$V$}\ .}
 \label{subsec:Vgeom}
 
 The vanishing of the supersymmetric variation of the gravitino imposes conditions on the bundle $V$.  More precisely, the curvature of the Yang-Mills connection satisfies
 \begin{align}
F\wedge\Psi  &= 0~,\label{eq:holF} \\
\omega\lrcorner F &= 0\label{eq:instV}~. 
 \end{align} 
The first condition is equivalent to $F^{(0,2)}= 0$, that is, $V$ {\it must be a holomorphic bundle}.  The second equation states that the curvature $F$ must be {\it primitive} with repect to $\omega$.  Both conditions together mean that  $V$ {\it must admit a  Hermitian Yang-Mills connection}.  Because the right hand side of equation \eqref{eq:instV} is zero, we say that the connection on  $V$ is an {\it instanton}.
We will be working with manifolds which are in general not K\"ahler, however we will take a moment to discuss the case where $X$ is K\"ahler.

When $X$ is K\"ahler, the existence  of a Hermitian Yang-Mills connection on $V$ is guaranteed by the work of Donaldson~\citep{0529.53018} and Uhlenbeck and Yau~\citep{MR861491, MR997570}.  We have the following theorem:
\vskip5pt

\begin{Theorem}[Donaldson, Uhlenbeck-Yau]
A polystable holomorphic vector bundle $V$ over a compact K\"ahler manifold $X$, with hermitian form $\omega$, admits a unique Hermitian Yang-Mills connection.
\end{Theorem}
\rightline{$\square$}
\noindent The stability refers to the slope $\mu(V)$ of $V$ which is defined as
$$\mu(V) = \frac{\int_X c_1(V)\wedge\omega^2}{\text{rk}(V)}~.$$
Stability of $V$ states that a vector bundle $V$ is {\it stable} if for all sub-sheaves  $E$ of $V$ with $0<\text{rk}(E)<\text{rk}(V)$ we have
$$\mu(E)<\mu(V)~.$$
A vector bundle $V = \oplus_i V_i$, is {\it polystable} if each $V_i$ is stable and satisfies $\mu(V_i)=\mu(V).$ 
Note that the Hermitian Yang Mills connection is unique, up to gauge transformations, for a given holomorphic structure on the bundle.

Buchdahl~\citep{MR939923} (for the case of complex surfaces) and,  Li and Yau~\citep{MR915839} (for higher dimensional complex manifolds) generalised this theorem to non-K\"ahler manifolds which admit a {\it Gauduchon metric}.  A Gauduchon metric $\hat g$ on a hermitian $n$ dimensional manifold with corresponding hermitian form $\widehat\omega$ is a metric that satisfies
$$\partial\bar\partial\,\widehat\omega{}^{n-1} = 0~.$$
For $n=3$, and a manifold $X$ which has an $SU(3)$ heterotic structure, this means that 
$$\widehat\omega = e^{-\phi}\omega$$ 
is Gauduchon, as it satisfies the balanced condition
$$\d( e^{-2\phi}\omega\wedge\omega)= \d(\widehat\omega\wedge\widehat\omega)=0~.$$
\vskip5pt

\begin{Theorem}[Buchdahl, Li-Yau]
Let $X$ be a compact hermitian manifold with a Gauduchon metric $\hat g$ and corresponding hermitian form $\widehat\omega$.  A polystable (with respect to the class $[\widehat\omega^2]$) holomorphic vector bundle $V$ over $X$ admits a unique Hermitian Yang-Mills connection.
\end{Theorem}
\rightline{$\square$}
\noindent The stability refers to the slope $\mu(V)$ of $V$ which is now defined as
\[\mu(V) = \frac{1}{{\text rk}(V)}\, \int_X c_1(V)\wedge\widehat\omega^2~,\]
and it states that for all sub-sheaves $E$ of $V$ it must be true that
\[ \mu(E)<\mu(V)~.\]
Therefore, when $X$ has a heterotic $SU(3)$ structure, we require the bundle $V$ to be a polystable (with respect to the class $[e^{-2\phi}\,\omega^2]$) holomorphic bundle, which thus guarantees the existence of Hermitian Yang-Mills connection on $V$.

For heterotic string compactifications, the relevant vector bundles have $c_1(V)=0$ and so {\it the slope vanishes} $\mu(V) =0$. In fact, for the gauge bundle, the gauge group is a subgroup of $E_8{\times} E_8$.  Also,
as we will see in the next section, we require that $TX$ be stable too, and so $\mu(TX)=0$ because 
$X$ has vanishing first Chern class.

\subsection{Constraints from the Anomaly Cancelation and Equations of Motion.}
\label{subsec:anomaly}

Apart from the constraints from supersymmetry, the pair $(X,V)$ must also satisfy an anomaly cancelation condition
\begin{equation}
H = \d B + {\cal C S}~,\label{eq:anomaly}
\end{equation}
where $B$ is a real 2 form, 
\[ {\cal C S} = \frac{\alpha'}{4}\, ({\rm CS}[A] - {\rm CS}[\Theta^I])~,\]
$A$ is the gauge connection for $V$, $\Theta^I$ is the connection for $TX$, and ${\rm CS}[A]$ and ${\rm CS}[\Theta^I]$ are the Chern--Simons 3-forms for these connections defined by
\begin{equation*} 
{\rm CS}[A] = \tr\left(A\wedge\d A + \frac{2}{3} A\wedge A\wedge A\right)~, 
\end{equation*}
and similarly for ${\rm CS}[\Theta^I]$.  The right hand side of the anomaly cancelation condition \eqref{eq:anomaly} is a definition of $H$ as the gauge invariant field strength of the $B$ field. The Bianchi identity for this anomaly cancelation condition is 
\begin{equation}
\d H = \frac{\alpha'}{4}\,\left(\tr (F\wedge F) - \tr (R^I\wedge R^I) \right) + {\cal W}~,\label{eq:BIanomaly}
\end{equation}
where $R^I$ is the curvature on $X$ with respect to a connection $\nabla^I$ on $TX$ which we discuss below.  The term $\cal W$ is a non-perturbative correction which is a closed 4-form on $X$ in a cohomology class which corresponds to the Poincar\'e dual of the class of an (effective) holomorphic curve $\cal{C}$ which is wrapped by a five-brane. 
A topological condition derives from equation \eqref{eq:BIanomaly}  
\begin{equation}
 0 = - \, {\cal P}_1(V) + {\cal P}_1(TX) + [\cal{C}]~,\label{eq:topanomaly}
 \end{equation}
 where ${\cal P}_1(E)$ represents the first Pontryagin class of a bundle $E$.  
In this paper we will ignore the non-pertubative correction $\cal W$.

Any solution $(X,V)$ of the supersymmetry conditions described in this section (that is, $X$ has a heterotic $SU(3)$-structure and $V$ is a poly-stable holomorphic bundle on $X$) which also satisfies the anomaly cancelation condition, automatically satisfies  the equations of motion if and only if the connection $\nabla^I$ satisfies
\begin{align} 
R^I\wedge\Psi &= 0~,\label{eq:holR}\\
 \omega\lrcorner R^I &= 0~.\label{eq:instR}
 \end{align}
That is, the connection $\nabla^I$ of the curvature $R^I$ of the tangent bundle $TX$ must be an $SU(3)$ instanton~\citep{Hull:1986kz, Ivanov:2009rh, Martelli:2010jx, delaossa2014}.  By the theorem of Li and Yau above, such a connection exists only if we require $(TX, \nabla^I)$ to be a stable holomorphic bundle on $X$.   We describe in more detail in the appendix and in a companion paper \citep{delaossa2014} the delicate issue of the choice of connection $\nabla^I$ in relation to the $\alpha'$ expansion.  We note here however that, to first order in the $\alpha'$ expansion, with the usual supersymmetry transformations, this connection is the Hull connection~\citep{Hull:1986kz}, and that it is easy to verify (see appendix) that this connection does satisfy equations \eqref{eq:holR} and \eqref{eq:instR} to this order.   In this paper, we take the connection $\nabla^I$ to be an instanton as in equations \eqref{eq:holR} and \eqref{eq:instR}.  Moreover, we will see that this connection needs to be promoted to a dynamical field to be able to understand the full moduli space of heterotic compactifications.

\subsection{Summary.}
\label{subsec:sum}

We are interested in the moduli of heterotic string compactifications which preserve $N=1$ supersymmetry.  In this paper we will refer to the pair $(X,V)$ as a {\it heterotic compactification} if it satisfies the {\it Strominger system of equations} as follows:
\begin{itemize}
\item $(X,\omega,\Psi)$ has a heterotic $SU(3)$ structure, that is
\begin{itemize}
\item $X$ is a complex manifold 
\item $X$ is conformally balanced: $\d(e^{-2\phi}\omega\wedge\omega) = 0$
\item $X$ has a holomorphically trivial canonical bundle:  $\d(e^{-2\phi}\, \Psi) = 0$
\end{itemize}
\item The bundle $V$ on $X$ must admit a connection $A$ with curvature $F$ taking values in $\End(V)$, which satisfies the hermitian Yang-Mills equations
\begin{equation*}
F\wedge\Psi  = 0~, \qquad
\omega\lrcorner F = 0~.
\end{equation*} 
Therefore, we require $V$ to be  a polystable holomorphic bundle. 
\item The bundle $TX$ on $X$ must admit a connection $\Theta^I$ with curvature $R^I$ taking values in $\End(TX)$, which satisfies the hermitian Yang-Mills equations
\begin{equation*}
R^I\wedge\Psi  = 0~, \qquad
\omega\lrcorner R^I = 0~.
\end{equation*} 
Therefore, we require $TX$ to be  a stable holomorphic bundle.
\item The flux $H$ (which is the torsion of the Bismut connection) and the connections $A$ and $\Theta^I$ must satisfy the anomaly cancelation condition which is
\[ H = i(\partial-\bar\partial)\,\omega = 
\d B + \frac{\alpha'}{4}\, \big({\rm CS}[A] - {\rm CS}[\Theta^I]\big)~.\]
Therefore the curvatures $F$ and $R^I$  and $H$ must satisfy the Bianchi identity
\[ \d H= - 2 i \p\bp \omega = \frac{\alpha'}{4}\,\left(\tr (F\wedge F) - \tr (R^I\wedge R^I) \right) ~.\] 
\end{itemize}
\newpage 

\section{Infinitesimal Moduli of Heterotic Compactifications.}
\label{sec:moduli}

In this section we study the space of {\it infinitesimal} deformations of a heterotic compactification $(X,V)$.  As described in detail in the following subsections, this moduli space contains the following parameters:

\begin{itemize}
\item Deformations of the {\it complex structure $J$ on $X$} (which is determined by $\Psi$).  It is well known that first order deformations $\Delta$ of the complex structure which preserve the integrability of the complex structure  belong to $H_{\bar\partial}^{(0,1)}(X, TX)$.  These deformations $\Delta$ are constrained by requiring that $\Omega$ stays holomorphic which in turn requires that deformations of the complex structure are in $H_\d^{(2,1)}(X)\subseteq H_{\bar\partial}^{(0,1)}(X, TX)$.  Moreover, they are also constrained by the requirement that the holomorphic conditions of the bundles $V$ ($F\wedge\Psi = 0$) and $TX$ ($R^I\wedge\Psi = 0$) be preserved. We also find a further constraint on $\Delta$ coming from the anomaly cancelation condition.  As the stability of both $V$ and $TX$ is not spoiled by first order deformations of $J$ which preserve the holomorphicity condition of both bundles, 
the theorem by Li and Yau guarantees that on the deformed heterotic compactification $(X',V')$ there is a connection on $V'$ and $(TX)'$ which satisfies the instanton equations.
\item Deformations of the bundle $V$ for a fixed complex structure $J$ and hermitian form $\omega$ on $X$, that is, those deformations of $V$ which are not accounted for by deformations of $\Psi$ and $\omega$.  
These are the {\it bundle moduli of $V$} which belong to $H_{\bar\partial_{\cal A}}^{(0,1)}(X, \End(V))$.  Similarly, we have deformations of the holomorphic tangent bundle $TX$, the {\it tangent bundle moduli}, which belong to $H_{\bar\partial_{\vartheta^I}}^{(0,1)}(X, TX))$.  Note that we are considering the instanton connection as an unphysical field in the theory.  We find that this is needed for the appropriate implementation of the anomaly cancelation condition, but we do not consider these moduli as corresponding to physical fields in the effective four dimensional field theory\footnote{See \citep{delaossa2014} for an extensive discussion.}.
\item Deformations of the hermitian structure $\omega$ which preserve the conformally balanced condition which are constrained by the anomaly cancelation condition.   
\end{itemize}

We leave the study of obstructions to these deformations for future work.

\subsection{First Order Deformations of Heterotic \texorpdfstring{$SU(3)$}\ \ Structures.}
 \label{subsec:infsu3}
 
 Let $(X, \omega, \Psi)$ be a manifold with a heterotic $SU(3)$ structure.  In this subsection we discuss first order variations of a heterotic $SU(3)$ structure.
   
Consider a one parameter family of manifolds $(X_t, \omega_t, \Psi_t), ~t\in\IC$, with a heterotic $SU(3)$ structure where we set
$(X, \omega,\Psi) = (X_0, \omega_0, \Psi_0)$. 
A deformation of the heterotic $SU(3)$ structure parametrized by the parameter $t$ corresponds to simultaneous deformations of the complex structure determined by $\Psi$ together with those of the hermitian structure
determined by $\omega$, such that the heterotic $SU(3)$ structure is preserved.  Hence the variation with respect to $t$ of any mathematical quantity $\alpha$ (as for example a $p$-form, or the metric) is given by the chain rule as follows
\begin{equation*}
\partial_t\alpha = (\partial_t z^a)\, \partial_a\alpha + 
(\partial_t z^{\bar a})\, \partial_{\bar a}\alpha + (\partial_t y^i)\, \partial_i\alpha ~,
\end{equation*}
where we label the complex structure parameters by $z^a$ and by $y^i$ the parameters of the hermitian structure\footnote{We will need to extend this later to include variations of the bundles.}

Note that $\Psi$ is independent of the hermitian structure parameters, however $\omega$ does depend on the complex structure as it must be  a $(1,1)$-form with respect to any complex structure. Therefore the moduli space ${\cal M}_X$ of the manifold $X$ must have the structure of a fibration.  We discuss this structure in the following sections.

\subsubsection{Variations of the complex structure of \texorpdfstring{$X$}\ .}
\label{sec:varsJ}

We begin this subsection by reviewing standard results on variations of an {\it integrable} complex structure $J$ of a manifold $X$.   With respect to $J$, the exterior derivative $\bp$ which acts on forms on $X$, squares to zero, that is $\bp^2=0$.  This condition is equivalent to the vanishing of the Nijenhuis tensor.  Conversely, a derivative $\bp$ which squares to zero defines an integrable complex structure on $X$. In fact, it determines a holomorphic structure on $X$.

Let $z^a, a = 1,\ldots, N_{CS}$, be complex structure parameters and $\Delta_a^m$ be a variation of the complex structure
\[ \Delta_a  = \Delta_a{}_{\, n}{}^m\, \d x^n\otimes \partial_m = - \frac{i}{2}\, \partial_aJ~.\]
It is a standard result that $\Delta_a^m\in \Omega^{(0,1)}(X, TX^{(1,0)})$. 
Further, preservation of the integrability of the complex structure under variations requires to first order that $\Delta_a^m$ defines an element of $H^{(0,1)}_{\bar\partial}(X,TX)$.  The integrability to first order is guaranteed (using the Maurer-Cartan equations) if $\bp\Delta_a = 0$, and $\bp$-exact forms $\Delta_a$ correspond to trivial changes of the complex structure (that is, holomorphic changes of the complex coordinates). 

Equivalently, as the form $\Psi$ on $X$ determines a unique integrable complex structure $J$ on a manifold with a heterotic structure $X$, one can study the  variations of $J$ in terms of the variations of $\Psi$. It will be more convenient however to discuss these deformations in terms of the holomorphic $(3,0)$ form.  Define
\[ \Omega = e^{-2\phi}\, \Psi~.\]
First order variations of $\Omega$ have the form~\citep{MR0112154, 0128.16902} 
\begin{equation} 
\partial_a \Omega = \widetilde K_a\, \Omega + \chi_a~,\label{eq:prevarO}
\end{equation}
where the $\widetilde K_a$ depend on the coordinates on $X$, and $\chi_a$ is a $(2,1)$-form which can be written in terms of $\Delta_a$
\begin{equation}
\chi_a = \half\, \Omega_{mnp}\, \Delta_a{}^m\wedge\d x^n\wedge\d x^p~.\label{eq:chi}
\end{equation}
Actually, we can prove that one can take the $\widetilde K_a$ to be  constants.

\begin{Proposition}\label{prop:30forms}
Let $\Lambda\in\Omega^{(3,0)}(X)$.  If $X$ has a holomorphically trivial canonical bundle with holomorphic form
$\Omega$, then
\[\Lambda = c\, \Omega + \p\lambda^{(2,0)}~,\]
for some constant $c$, and $(2,0)$-form $\lambda^{(2,0)}$ .
\end{Proposition}

\begin{proof}
The Hodge decomposition of $\Lambda$ with respect to the $\p$-operator is
\[ \Lambda = \p\lambda^{(2,0)} + \Lambda^{\p-har}~,\]
for some $(2,0)$-form $\lambda^{(2,0)}$ and $\p$-harmonic $(3,0)$-form $\Lambda^{\p-har}$. It is easy to see that $\Lambda^{\p-har}$ must be holomorphic.  Computing
\[ 0 = \p^\dagger(\Lambda^{\p-har}) = - * \bp * (\Lambda^{\p-har})
= i\, * \bp (\Lambda^{\p-har}) ~,\]
where we have used 
\[ * (\Lambda^{\p-har}) = J(\Lambda^{\p-har}) = -i  (\Lambda^{\p-har})~,\]
we find 
\[ * \bp (\Lambda^{\p-har}) = 0~.\]
This implies that
\[ \bp (\Lambda^{\p-har}) = 0~,\]
and therefore, that $\Lambda^{\p-har}$ is a holomorphic $(3,0)$-form.  But as $\Omega$ is unique
we obtain
\[\Lambda = c\, \Omega + \p\lambda^{(2,0)}~,\]
where $c$ is a constant.
\end{proof}
Using the proposition, we can now write equation \eqref{eq:prevarO} as
\begin{equation}
\partial_a \Omega = K_a\, \Omega + \chi_a~,\label{eq:varO}
\end{equation}
where now the $K_a$ are constants and we have ignored the $\p$-exact term.  This term can be ignored because it corresponds to changes in $\Omega$ due to diffeomorphisms of $X$, that is, trivial deformations of the complex structure.  This can be seen by computing the  Lie-derivative of $\Omega$ along a vector $V\in TX$ which gives
\begin{equation}
{\cal L}\Omega = -\delta_{diff}\,\Omega = \d(v\lrcorner\Omega)~,\label{eq:Odiffeo}
\end{equation}
where we have used the fact that $\d\Omega= 0$.
Taking the $(3,0)$ part of this equation, we obtain
\[ ({\cal L}\Omega)^{(3,0)} = \p(v\lrcorner\Omega)~,\]
which is precisely a $\p$-exact $(3,0)$-form, and any $(2,0)$-form can be written as $v\lrcorner\Omega$
for some $v\in TX$.

The integrability of the deformed complex structure given by equation \eqref{eq:varO} is obtained by 
varying the equation
\[\d \Omega = 0~,\]
and demanding that the deformed manifold admits a holomorphic $(3,0)$-form.  We find
\begin{equation}
 \d\chi_a = - \d K_a\wedge\Omega = 0~.
 \label{eq:dchi}
 \end{equation}
 Therefore  each $\chi_a$ defines a class in the de-Rham cohomology
 \[ \chi_a \in H_{\d}{}^{(2,1)}(X)~,\]
as $\d$-exact $(2,1)$-forms correspond to 
diffeomorphisms of $X$, as can be seen from equation \eqref{eq:Odiffeo}.

We remark that using the holomorphicity of $\Omega$, and equation \eqref{eq:varO}, it is straightforward to prove that $\chi_a\in H_{\bp}{}^{(2,1)}(X)$ is equivalent to $\Delta_a{}^m \in H^{(0,1)}_{\bar\partial}(X,TX)$.  In fact, $\Omega$ gives an isomorphism between these cohomology groups (just like in the case of Calabi--Yau manifolds)~\citep{MR915841}
\[H_{\bp}{}^{(2,1)}(X)\cong H^{(0,1)}_{\bar\partial}(X,TX)~.\] 
However, on a non-K\"ahler manifold with a holomorphically trivial canonical bundle, it is not necessarily the case that
$H_{\d}{}^{(2,1)}(X) \cong H_{\bp}{}^{(2,1)}(X)$
and it is generally the case that 
\[ {\rm dim} H_{\d}{}^{(2,1)}(X) \le {\rm dim} H_{\bp}{}^{(2,1)}(X)~.\]
The way to see this is by observing that first order variations of  \eqref{eq:dchi} require, not only that $\chi_a$ is $\bp$-closed, but also that it is $\p$-closed. Therefore, in a given class of $[\chi_a]\in H_{\bp}{}^{(2,1)}(X)$, there must exist a representative which is  $\p$-closed and is not $\d$-exact.  This is not always the case, and there are many examples of non-K\"ahler manifolds for which this happens.  A simple example with a heterotic $SU(3)$ structure is the Iwasawa manifold\footnote{However, the Iwasawa manifold is not a good heterotic compactification for any bundle $V$ because its holomorphic tangent bundle is not stable.  Note in particular that this bundle is holomphically trivial and so ${\rm dim}(H^0(X,TX)) = 3$, which implies that there is no connection $\nabla^I$ for which $TX$ is stable.}.  It is not too hard to show that for this example~\citep{MR0393580}
\[ {\rm dim} H_{\d}{}^{(2,1)}(X) = 4~,\quad{\rm and} \quad {\rm dim} H_{\bp}{}^{(2,1)}(X) = 6~.\]
The two extra elements in ${\rm dim} H_{\bp}{}^{(2,1)}(X)$ are $\p$-closed, however they are $\d$-exact. 

There are also many examples of non-K\"ahler manifolds for which 
\begin{equation} H_{\d}{}^{(2,1)}(X)\cong H_{\bp}{}^{(2,1)}(X)~,
\label{eq:CScong}
\end{equation}
like for example manifolds which are {\it cohomologically K\"ahler}, that is, which satisfy the $\p\bp$-lemma, a property which is stable under complex structure deformations~\citep{MR2451566,MR2449178}.  Note that the Iwasawa manifold does not satisfy the $\p\bp$-lemma.

The condition that each $\chi\in H_{\bp}{}^{(2,1)}(X)$ also satisfies $\p\chi=0$ is used  in~\citep{MR915841} as a first step to discuss the obstructions to first order deformations of the complex structure $J$ of Calabi--Yau manifolds, and it is stated in that proof that it goes through even if  the manifold is not  K\"ahler, as long as it satisfies the $\p\bp$-lemma.  In our work, the requirement that $\p\chi=0$ appears at first order in deformation theory when discussing the deformations of the complex structure in terms of the variations of $\Omega$ and {\it it represents a necessary condition for integrability to first order}.  Issues including  the integrability of the deformed complex structure of $X$ and a generalisation of the work of Tian and Todorov~\cite{MR915841,MR1027500}, are  discussed in a forthcoming paper~\citep{DKS2014}.  For the rest of this paper, we work with the variations of the complex structure in terms of $\Delta$, but we should keep in mind that some of these elements may be obstructed as discussed.

\subsubsection{Deformations of the hermitian structure on \texorpdfstring{$X$}\ .}
\label{subsubsec:defHS}

Recall that a manifold with a heterotic structure $X$ has a hermitian form given by $\omega$, and that $\omega$ satisfies
\[ \d (e^{-2\phi}\omega\wedge\omega) = 0~,\]
that is, $X$ is conformally balanced.  This equation varies with respect to the complex structure (because $\omega$ is a $(1,1)$ form) and the hermitian structure. 

Let
\[\hat\rho = \half\, \widehat\omega\wedge\widehat\omega~,\]
where $\widehat\omega = e^{-\phi}\omega$ is the Gauduchon metric.
The conformally balanced condition is equivalent to
\[\d \hat\rho = 0~,\]
and so $\hat\rho\in H_\d^4(X)$. Any variation of $\hat\rho$ must preserve this condition, that is
\[ \d (\partial_t\hat\rho) = 0.\]
Consider the action of  a diffeomorphism of $X$ on $\hat\rho$
\[{\cal L}_{diff} \hat\rho = \d (v\lrcorner\hat\rho) = - \d( e^{-\phi} \, J(v)\wedge\widehat\omega)~.\]
Therefore, variations of $\hat\rho$ which preserve the conformally balanced condition correspond to $\d$-closed four forms modulo 
$\d$-exact forms which have the form $\d\beta$, where $\beta$ is a {\it non-primitive} three-form.  So this space is not necessarily finite dimensional as was first pointed out in~\citep{Becker:2006xp}. 

As we will see,  when taking into account the anomaly cancelation condition, we obtain a finite dimensional parameter space. For the remainder of this section, we set up some notation and make some further remarks on the deformations of hermitian structure of $X$.

Consider a variation, $\p_t\omega$, of $\omega$.  We can decompose this variation in terms of the Lefshetz decomposition
\begin{equation}
 \p_t\omega = \lambda\omega + h_t~,\label{eq:varomega}
 \end{equation}
where $\lambda$ is a function on $X$ and $h_t$ is a primitive two form ($\omega\lrcorner h_t = 0$).

It is not too difficult to show that the $(0,2)$-part of the variation, $h_t^{(0,2)}$, depends only on the variations of the complex structure $\Delta_a$.  To prove this we vary the compatibility condition \eqref{eq:compatibility} 
\[ \Omega\wedge\omega = 0~,\]
which expresses the fact that with respect to the complex structure $J$ determined by the $(3,0)$-form $\Psi = e^{2\phi}\, \Omega$, the hermitian form $\omega$ is a $(1,1)$-form.  
Varying equation \eqref{eq:compatibility} and using \eqref{eq:varO}, we find
\[0= \partial_t\omega\wedge\Omega + \omega\wedge\partial_t\Omega = \partial_t\omega\wedge\Omega+ \omega\wedge\chi_t~,\]
where
\[ \Delta_t = (\partial_t z^a)\,\Delta_a~, 
\quad{\rm and}\quad \chi_t = (\partial_t z^a)\,\chi_a~.\]
Contracting with $\bar\Omega$ we obtain
\begin{equation*}
  h_t^{(0,2)} = (\partial_t z^a)\,h_a^{(0,2)}
 \end{equation*}
 where
 \begin{equation}
 h_a^{(0,2)} = \Delta_a{}^m\wedge \omega_{mn}\, \d x^n~,
 \label{eq:varpureomega}
 \end{equation}
and where we have used equation \eqref{eq:chi}.  Therefore, the $(0,2)$-part of the variation of $\omega$ is entirely  determined by 
the allowed variations of the complex structure of $X$, and there are no new moduli associated to $h_a^{(0,2)}$.

We would like to remark that it has been known for over 20 years in mathematics that the conformally balanced condition is not stable under deformations of the complex structure \citep{MR1107661, MR1137099, 0735.32009, 0793.53068}.  This is in sharp contrast with the theorems of Kodaira and Spencer for the stability of the K\"ahler condition under deformations of the complex structure \citep{MR0112154, 0128.16902}.
In fact, we have that under deformations of the complex structure alone (see also \citep{Becker:2005nb} where the authors consider deformations of the complex structure only)
\[\d(\p_a\hat\rho) = \d(\widehat\omega\wedge\p_a\widehat\omega) 
= \d(\Delta_a{}^m\wedge\widehat\omega\wedge \widehat\omega_{mn} \d x^n) = 0~,\]
which seems to be a difficult equation to satisfy.

Returning now to the variation of the conformally balanced condition for the hermitian structure we have the following proposition.

\begin{Proposition}\label{prop:CBcond}
Let $\widehat\lrcorner$ and $\hat *$ be the contraction operator and the Hodge dual operator with respect to 
the Gauduchon metric respectively.
The variation of the conformally balanced condition for the hermitian structure
\begin{equation}
\d \p_t\hat\rho = 0~,\label{eq:dvarhatrho}
\end{equation}
determines the $\bp$-exact part of Hodge decomposition of the $(1,1)$-form 
\[\hat*\, (\p_t\hat\rho)^{(2,2)} = - \hat h_t{}^{(1,1)} + 2\hat\lambda_t\, \widehat\omega~,\] 
in terms of deformations of the complex structure leaving the $\bp^{\hat\dagger}$-closed part undetermined, as long as we assume that 
the tangent bundle is stable and has zero slope.   
\end{Proposition}

\begin{proof}
From equation \eqref{eq:dvarhatrho} we get
\begin{equation*} 
 \d^{\widehat\dagger} (2\hat \lambda_t \,\widehat\omega - \hat h_t^{(1,1)} + \hat h_t^{(0,2)} + \hat h_t^{(2,0)}) = 0~,
 \end{equation*} 
where $d^{\widehat\dagger}$ is the adjoint of the exterior derivative $\d$ with respect to the Gauduchon metric, and where we have used equations \eqref{eq:varomega}, and $\hat h_t = e^{-\phi}\, h_t$ and $\hat\lambda_t = \lambda_t - \p_t\phi$.  
Consider the $(1,0)$ part of this equation
\begin{equation}
\p^{\hat\dagger} \hat h_t{}^{(2,0)} = \bp^{\hat\dagger}( - 2\hat\lambda_t\, \widehat\omega + \hat h_t{}^{(1,1)})~.
\label{eq:eqforh}
\end{equation}
On a manifold with a stable tangent bundle and zero slope, one can prove that
the left hand side of this equation is  $\p^{\hat\dagger}$-coexact because
\[ H_{\bp}^{(2,0)}(X)\cong H_{\bp}^0(X, TX) = 0~.\]
The last equality follows from slope-zero stability, and the isomorphism of cohomologies is due the $\Omega$ isomorphism (that is,
for every element in $s^m\in H_{\bp}^0(TX)$ we have an element in $H_{\bp}^{(2,0)}(X)$ given by $s^m\,\Omega_{mnp}$).
The Hodge decomposition of $\hat h_t^{(2,0)}$ in terms of the Laplacian $\widehat\Delta_{\bp}\,$ requires that
\begin{equation*}
\hat h_t^{(2,0)}=\bp^{\hat\dagger}\Lambda_t^{(2,1)},
\end{equation*}
for some $(2,1)$-form $\Lambda$. Recall that $\hat h_t^{(2,0)}$ is completely determined by the complex structure deformations (see equation \eqref{eq:varpureomega}). It follows that $\Lambda^{(2,1)}$ is also given in terms of complex structure variations.  Equation \eqref{eq:eqforh} can now be written as
\begin{equation*}
\bp^{\hat\dagger} (\hat h_t{}^{(1,1)} - 2\hat\lambda_t\, \widehat\omega ) 
= - \bp^{\hat\dagger} ( 
 \p^{\hat\dagger}\Lambda_t^{(2,1)})~,
\end{equation*}
which means that the left hand side is entirely determined by variations of the complex structure.  Moreover, using the Hodge decomposition,
we find that this equation determines, as claimed, the $\bp$-exact part of the $(1,1)$-form 
\[\hat*\, (\p_t\hat\rho)^{(2,2)} = - \hat h_t{}^{(1,1)} + 2\hat\lambda_t\, \widehat\omega~,\] 
in terms of deformations of the complex structure, and leaves the $\bp^{\hat\dagger}$-closed part undetermined. 
\end{proof}

By enforcing the anomaly cancelation condition, we will able to fix these parameters further. 
In fact, we will argue later in section \ref{subsec:Anomaly} that, by enforcing the anomaly cancelation condition, the moduli space of the (complexified) hermitian form is finite dimensional and related to the cohomology group
 \begin{equation*}
 H^{(0,1)}_{\bp}(X,T^{*(1,0)}X)\:.
\end{equation*} 

Finally, before continuing with our analysis of the moduli space of the Strominger system,  
we would like to point out that in~\citep{DKS2014} we show that one can consider a one parameter family of manifolds $(X_t, \Psi_t, \omega_t)$ with a heterotic structure such that, for $t\in \IR$, the family has an integrable $G_2$ structure or, for $t\in \IC$, the family has a certain $SU(4)$ or $\text{Spin}(7)$ structure.  Requiring that the family admits one of these $G$-structures guarantees that the heterotic structure, and hence the conformally balanced condition, is satisfied. Conversely, one can construct manifolds with one of these $G$ structures which have embedded a family manifolds with a heterotic $SU(3)$ structure.  We find that this is very interesting for applications to $F$-theory and $M$-theory.   

\subsection{Variations of the Holomorphic Structure on \texorpdfstring{$V$}\ .}
\label{subsec:defsV}

In this subsection, we study deformations  of the holomorphic structure of $V$.  The study of deformations of the holomorphic bundles has a long history in mathematics.  In this section, of particular relevance is  the work of Atiyah~\citep{MR0086359} which considers the parameter space of simultaneous deformations of the complex structure on a manifold $X$ together with those of the holomorphic structure on $V$.  This work has already been applied to the case in which $X$ is a Calabi-Yau manifold~\citep{Anderson:2010mh, Anderson:2011ty}, and in this section we extend it to the more general case of a manifold with a heterotic $SU(3)$ structure. We will do this in detail, even though not much is different for this part of the parameter space, as it is the structure that we encounter here that generalises when we include the more complicated anomaly cancelation condition.

Consider now a one parameter family of heterotic compactifications $(X_t, V_t)~t\in\IC$
where we set  $(X_0, V_0) = (X, V)$.
We study simultaneous deformations of the complex structure determined by $\Psi$ and the holomorphic structure on $V$.  Hence the variation with respect to $t$ of any mathematical quantity $\beta$ (which may have values in $V$ or $\End V$) is given by the chain rule as follows
\begin{equation*}
\partial_t\beta = (\partial_t z^a)\, \partial_a\beta+ 
(\partial_t z^{\bar a})\, \partial_{\bar a}\beta + (\partial_t y^i)\, \partial_i\beta
+ (\partial_t \lambda^\alpha)\, \partial_\alpha\beta
+ (\partial_t \lambda^{\bar\alpha})\, \partial_{\bar\alpha}\beta
\end{equation*}
where we label the bundle moduli by $\lambda^\alpha$.

Let $F$ be the curvature of the bundle $V$ where
\begin{equation}
F = \d A + A\wedge A~,\label{eq:curvV}
\end{equation}
and where $A\in\Omega^1(X,\End(V))$ is the gauge potential.  
Let $\beta\in\Omega^{(0,q)}(X, \End V)$\footnote{We only need to work with $(0,q)$ forms, however our work generalises to any $(p,q)$ forms.}.   We can define an exterior derivative on
$V$ by
\begin{equation} \d_A = \d + [A, ]~,
\label{eq:dA} 
\end{equation}
where, $[A, ]$ is defined by
\[ [A, \beta] = A\wedge\beta - (-1)^q\beta\wedge A~.\]
A holomorphic structure on $V$ is determined by the derivative $\bar\partial_{\cal A}$ which is defined as the $(0,1)$ part of the operator $\d_A$, that is,
\begin{equation} \bar\partial_{\cal A} \beta = \bp\beta + [{\cal A}, \beta]~,
\label{eq:bpA}
\end{equation}
where $\cal A$ is the $(0,1)$ part of $A$.  It is easy to prove that $\bp_{\cal A}{}^2 = 0$ only if $F^{(0,2)} = 0$.  

Consider now what happens to the holomorphicity of the bundle $V$ under deformations of the complex structure of $X$.   
Varying equation \eqref{eq:holF} and using \eqref{eq:varO}, we find
\[0= \partial_a F\wedge\Psi + F\wedge\partial_a\Psi = \partial_a F\wedge\Psi + F\wedge\chi_a~.\]
Therefore
\begin{equation}
 (\partial_a F)^{(0,2)} = \Delta_a{}^m\wedge F_{mn}\, \d x^n~,\label{eq:varpureFone}
 \end{equation}
where we have used equation \eqref{eq:chi}.
On the other hand, varying \eqref{eq:curvV} we find that
\begin{equation}
 (\partial_a F)^{(0,2)} = \bar\partial_{\cal A}\,\alpha_a~,\label{eq:varpureFtwo}
 \end{equation}
 where $\alpha_a$ is the $(0,1)$ part of the variation of $A$.   Putting together equations \eqref{eq:varpureFone} and \eqref{eq:varpureFtwo} we find
 \begin{equation}
 \bar\partial_{\cal A}\,\alpha_a = \Delta_a{}^m\wedge F_{mn}\, \d x^n~.\label{eq:atiyahF}
 \end{equation}
 This equation represents a constraint on the possible variations $\Delta_a$ of the complex structure $J$ on $X$.
 
Consider the map 
\begin{equation}
\label{eq:mapF}
{\cal F}\;:\;\Omega^{(0,q)}(X,T^{(1,0)}X)\longrightarrow \Omega^{(0,q+1)}(X,\End(V))
\end{equation}
given by
\begin{equation}
{\cal F}\big(\Delta\big)= (-1)^q\, \Delta{}^m\wedge F_{mn}\, \d x^n~,
\qquad \Delta\in\Omega^{(0,q)}(X,T^{(1,0)}X)~.
\label{eq:mapFdef}
\end{equation}
We have the following theorem: 
\begin{Theorem}\label{prop:one}
\begin{equation}
\bp_{\cal A} \left({\cal F}\big(\Delta\big)\right) = - {\cal F}\left(\bp\Delta\right)~,
\quad\forall\, \Delta\in H_{\bp}^{(0,q)}(X,T^{(1,0)}X)~,
\label{eq:Fcohom}
\end{equation}
and therefore the map ${\cal F}$ is a map between cohomologies
\begin{equation}
\label{eq:mapFcohom}
{\cal F}\;:\;H_{\bp}^{(0,q)}(X,T^{(1,0)}X)\longrightarrow H_{\bp_{\cal A}}^{(0,q+1)}(X,{\rm End}(V))~.
\end{equation}
\end{Theorem}

\begin{proof}
Using equation \eqref{eq:bpA}
\begin{align*}
\bp_{\cal A} \left({\cal F}\big(\Delta\big)\right)&=
\bp\left({\cal F}\big(\Delta\big)\right) + {\cal A}\wedge {\cal F}(\Delta) - (-1)^{q+1}\, {\cal F}(\Delta)\wedge{\cal A}
\\
&= (-1)^q\ \bp\left(\Delta^m\wedge F_{mn}\d x^n\right) +
{\cal A}\wedge {\cal F}(\Delta) + (-1)^q\, {\cal F}(\Delta)\wedge{\cal A}
\\
&= (-1)^q\ \bp\left(\Delta^m\right)\wedge F_{mn}\d x^n +
\Delta^m\wedge \bp\left(F_{mn}\d x^n\right) \\
&\qquad
+ {\cal A}\wedge {\cal F}(\Delta) + (-1)^q\, {\cal F}(\Delta)\wedge{\cal A}\\[5pt]
&= - \, {\cal F}\left(\bp\Delta\right) +
\Delta^m\wedge\left( 
\bp\left(F_{mn}\d x^n\right)  +  {\cal A}\wedge F_{mn}\d x^n+ F_{mn}\d x^n\wedge{\cal A}
\right)\\
&= - \, {\cal F}\left(\bp\Delta\right) +
\Delta^m\wedge \bp_{\cal A}\left(F_{mn}\d x^n\right)~. 
\end{align*}
The last term vanishes for every $\Delta^m\in \Omega^{(0,1)}(X, T^{(0,1)})$ due to the Bianchi identity for the curvature $F$
\[\bp_{\cal A} F = 0~.\]
In fact, this Bianchi identity implies that
\begin{equation}
P_{m}{}^p\,\bp_{\cal A} \big(F_{pn}\d x^n\big)~ = 0,  \label{eq:strongBIF}
\end{equation}
where $P$ and $Q$ are the projection operators
\[ P = \half\, (1 -  i J)~,\qquad Q = \half\, (1 +  i J)~.\]
Thus
\[ \Delta^m\wedge\bp_{\cal A} \big(F_{mn}\d x^n\big)~ = 0~,\quad\forall\ 
\Delta^m\in \Omega^{(0,1)}(X, T^{(0,1)})~.\]
Therefore 
we have proven equation \eqref{eq:Fcohom}, which also implies that
\[\bp\Delta = 0\qquad\Longrightarrow \qquad \bp_{\cal A}\left({\cal F}\big(\Delta\big)\right) = 0~.\]
and so $\cal F$ is a map between cohomologies as in equation \eqref{eq:mapFcohom}.
\end{proof}

\vskip10pt
\noindent We will refer to the map $\cal F$ as the {\it Atiyah map for $F$}.  It is worth remarking that this map is well defined as a map between cohomologies.  In fact, as under gauge transformations the curvature $F$ is covariant, then so  is ${\cal F}(\Delta)$.  Therefore, equation \eqref{eq:Fcohom} is invariant under gauge transformations.
\vskip5pt

In terms of the map $\cal F$, the constraint \eqref{eq:atiyahF} on the variations of the complex structure 
$\Delta _a\in H_{\bp}^{(0,1)}(X,T^{(1,0)}X)$
can now be written as
\begin{equation}
 \bar\partial_{\cal A}\,\alpha_a = - {\cal F}\big(\Delta_a\big)~.\label{eq:atiyahFbis}
 \end{equation}
So  ${\cal F}\big(\Delta_a\big)$ must be exact in $H_{\bp_{\cal A}}^{(0,2)}(X,\End(V))$, in other words
\[ \Delta _a\in \ker({\cal F})\subseteq H_{\bp}^{(0,1)}(X,T^{(1,0)}X)~.\]

The tangent space ${\cal TM}_1$ of the moduli space of combined deformations of the complex structure and bundle deformations, keeping fixed the hermitian structure, is given by
\begin{equation}
{\cal TM}_1= H_{\bp_{\cal A}}^{(0,1)}(X,\End(V))\oplus \ker({\cal F})~ ,\label{eq:M1}
\end{equation}
where $H_{\bp_A}^{(0,1)}(X,\End(V))$ is the space of bundle moduli. 

These results can be restated in a way that will be suitable for generalisations later when we include the other constraints on the heterotic compactification $(X,V)$.  Define a bundle ${\cal Q}_1$ which is the extension of $TX$ by $\End(V)$, given by the short exact sequence
  \begin{equation}
   0\rightarrow \End(V)\xrightarrow{\iota_1} \Q_1 \xrightarrow{\pi_1} TX \rightarrow 0~,   \label{eq:ses1}\end{equation}
 with extension class $\cal F$.
 There is a holomorphic structure on ${\cal Q}_1$ defined by the exterior derivative $\bar\partial_1$ 
 \begin{equation*}
\bp_1\;=\; \left[ \begin{array}{cc}
\bp_{\cal A} & \;{\cal F} \\
0  & \bp  \end{array} \right],
\end{equation*}
which acts on $\Omega^{(0,q)}(\Q_1)$ and squares to zero, $\bp_1^2 = 0$.  In fact, we have

\begin{Cor}\label{prop:two}
\[\bp_1^2 = 0~.\]
\end{Cor}
\begin{proof}
Let
\begin{equation*}
\left(
\begin{array}{c}
\alpha\\ \Delta
\end{array}\right)
\in \Omega^{(0,q)}(X,{\rm End}(V)) \oplus \Omega^{(0,q)}(X, TX)~.
\end{equation*}
Then 
\begin{equation*}
\bp_1^2 
\left(
\begin{array}{c}
\alpha\\ \Delta
\end{array}\right) =
\left[ 
\begin{array}{c}
\bp_{\cal A}{\cal F}\big(\Delta\big) +  {\cal F}\big(\bp\Delta\big) \\
0    \end{array} 
\right] = 0~,
\end{equation*}
by theorem \ref{prop:one}.
\end{proof}
\vskip10pt
\noindent We remark that $\bp_1^2=0$ is due to the Bianchi identity for the curvature $\bp_{\cal A} F =  0$. 

The infinitesimal moduli space of the holomorphic structure $\bp_1$ on the extension bundle $\Q_1$, which is given by
\[{\cal TM}_1 = H^{(0,1)}_{\bp_1}(X, {\cal Q}_1)~,\]
can be computed by a long exact sequence in cohomology (for more details see~\cite{Anderson:2010mh})
\begin{equation} 
\begin{split}
0 &\rightarrow H^{(0,1)}(\End(V)) \xrightarrow{\iota_1'} H^{(0,1)}(\Q_1) \xrightarrow{\pi_1'} H^{(0,1)}(TX) \\
&\xrightarrow{\cal F} H^{(0,2)}(\End(V)) \rightarrow H^{(0,2)}(\Q_1) \rightarrow \ldots
\label{eq:les1}
\end{split}
\end{equation}
where the Atiyah map $\cal F$ is the connecting homomorphism as can be deduced from theorem \ref{prop:one}. Another way to see that the connecting homomorphism is given by the extension class $\cal F$ is from its definition
\begin{equation}
[\iota_1^{-1}\circ\bp_1\circ \pi_1^{-1}(x)]=[{\cal F}(\Delta)]\:.
\end{equation}
where we have used the definition of $\bp_1$ above.  In the computation of the long exact sequence \eqref{eq:les1},  we have used
\[H_{\bp}^0(X,TX) = 0~,\]
because $\mu(TX)= 0$ and we require $TX$ to be a stable bundle. Thus, we also have
\begin{equation}
 H^0(X, \Q_1) \cong H^0(X, \End(V))~.\label{eq:nosecQ1}
 \end{equation} 
Recall that for a stable bundle $V$   
\[ {\rm dim} H^0(X, \End(V))\le 1~.\] 
There are non-trivial sections whenever the trace of the endomorphisms is non-vanishing. Then,  for a polystable bundle
\begin{equation*}
V=\oplus_{i=1}^nV_i~,
\end{equation*}
we have 
\[\textrm{dim}(H^0(X, \End(V)))=\tilde n-1~\]
where $\tilde n$ is the number of bundle factors which have  endomorphisms non-vanishing trace, and we subtract one as the overall trace should vanish.

Finally, we find, by exactness of the sequence \eqref{eq:les1}, that
\[{\cal TM}_1= H_{\bp_1}^{(0,1)}(X, \Q_1) = {\rm Im}(\iota_1') \oplus {\rm Im}(\pi_1') \cong
H_{\bp_A}^{(0,1)}(X,\End(V))\oplus \ker({\cal F}) ~,\]
in agreement with equation \eqref{eq:M1}.

\subsection{Variations of the Holomorphic Structure on \texorpdfstring{$TX$}.}
\label{subsec:defsTX}

We now extend our results to include deformations of the holomorphicity condition \eqref{eq:holR} of the tangent bundle $TX$ under deformations of the complex structure of $X$.   Basically, we repeat the analysis above.  Let $R^I$ be the curvature of the tangent bundle 
\begin{equation}
R^I = \d \Theta^I + \Theta^I\wedge \Theta^I~,\label{eq:curvR}
\end{equation}
and where $\Theta^I\in\Omega^1(X,\End(TX))$ is the tangent bundle instanton connection.  Let $\beta\in\Omega^{(0,q)}(X, TX)$.   We define an exterior derivative on $TX$ by
\[ \d_{\Theta^I}\beta = \d\beta + [\Theta^I,\beta]~.\]
A holomorphic structure on $TX$ is determined by the derivative $\bar\partial_{\vartheta^I}$ which is defined as the $(0,1)$ part of the operator $\d_{\Theta^I}$, that is,
\begin{equation} \bar\partial_{\vartheta^I} \beta = \bp\beta + [{\vartheta^I},\beta]~,
\label{eq:bptheta}
\end{equation}
where $\vartheta^I$ is the $(0,1)$ part of $\Theta^I$.  It is easy to prove that $\bp_{\vartheta^I}{}^2 = 0$ only if $R^I{}^{(0,2)} = 0$.  
Varying equation \eqref{eq:holR} and using \eqref{eq:varO}, 
\begin{equation}
 (\partial_a R^I)^{(0,2)} = \Delta_a{}^m\wedge R^I{}_{mn}\, \d x^n~,\label{eq:varpureRone}
 \end{equation}
where we have used equation \eqref{eq:chi}.
On the other hand, varying \eqref{eq:curvR} we find that
\begin{equation}
 (\partial_a R^I)^{(0,2)} = \bar\partial_{\vartheta^I}\,\kappa_a~,\label{eq:varpureRtwo}
 \end{equation}
 where $\kappa_a$ is the $(0,1)$ part of the variation of $\Theta^I$.   Putting together equations \eqref{eq:varpureFone} and \eqref{eq:varpureRtwo} we find
 \begin{equation}
 \bar\partial_{\theta^i}\,\kappa_a = \Delta_a{}^m\wedge R^I_{mn}\, \d x^n~.\label{eq:atiyahR}
 \end{equation}
 This equation represents a further constraint on the possible variations $\Delta_a$ of the complex structure $J$ on $X$.
  
Consider the map 
\begin{equation}
\label{eq:mapR}
{\cal R}^I\;:\;\Omega^{(0,q)}(X,T^{(1,0)}X)\longrightarrow \Omega^{(0,q+1)}(X, {\rm End}(TX))
\end{equation}
given by
\begin{equation}
{\cal R}^I\big(\Delta\big)= (-1)^q\, \Delta{}^m\wedge R^I_{mn}\, \d x^n~,
\qquad \Delta\in\Omega^{(0,q)}(X,T^{(1,0)}X)~.
\label{eq:mapRdef}
\end{equation}
We have the following theorem: 
\begin{Theorem}\label{prop:three}
\begin{equation}
\bp_{\vartheta^I} \left({\cal R}^I\big(\Delta\big)\right) = - {\cal R}^I\left(\bp\Delta\right)~,
\quad\forall\, \Delta\in H_{\bp}^{(0,q)}(X,T^{(1,0)}X)~,
\label{eq:Rcohom}
\end{equation}
and therefore the map ${\cal R}^I$ is a map between cohomologies
\begin{equation}
\label{eq:mapRcohom}
{\cal R}^I\;:\;H_{\bp}^{(0,q)}(X,T^{(1,0)}X)\longrightarrow H_{\bp_{\vartheta^I}}^{(0,q+1)}(X,{\rm End}(TX))~.
\end{equation}
\end{Theorem}
\begin{proof}
The proof is just like that for theorem \ref{prop:one} and follows from the Bianchi identity
\[ \bp_{\vartheta^I} R^I = 0~.\]
\vskip-10pt
\end{proof}
\noindent We will refer to the map ${\cal R}^I$ as the {\it Atiyah map for $R^I$}.
We remark that this map is also well defined as a map between cohomologies because equation \eqref{eq:Rcohom} is invariant under gauge transformations.  

\vskip5pt

In terms of the map ${\cal R}^I$, the constraint \eqref{eq:atiyahR} on the variations of the complex structure 
$\Delta _a\in H_{\bp}^{(0,1)}(X,T^{(1,0)}X)$
can now be written as
\begin{equation}
 \bar\partial_{\vartheta^I}\,\kappa_a = - {\cal R}^I\big(\Delta_a\big)~,\label{eq:atiyahRbis}
 \end{equation}
so  ${\cal R}^I\big(\Delta_a\big)$ must be exact in $H_{\bp_{{\theta}^I}}^{(0,2)}(X,TX)$, in other words
\[ \Delta _a\in \ker({\cal R}^I)\subseteq H_{\bp}^{(0,1)}(X,T^{(1,0)}X)~.\]

The tangent space of the moduli space ${\cal T M}_2$ of allowed combined deformations of the complex structure, bundle deformations and tangent bundle deformations, keeping fixed the hermitian structure, is now given by
\begin{equation}
{\cal T M}_2= H_{\bp_{\vartheta^I}}^{(0,1)}(X,{\rm End}(TX))\oplus
H_{\bp_ {\cal A}}^{(0,1)}(X,\End(V))\oplus \left(\ker({\cal F})\cap\ker({\cal R}^I)\right) ~ ,
\end{equation}
where $H_{\bp_{\vartheta^I}}^{(0,1)}(X,{\rm End}(TX))$ is the space of deformations of the connection $\nabla^I$ on tangent bundle $TX$. \vskip5pt

These results can be restated in terms of an extension $E$ of the bundle ${\cal Q}_1$. Define a bundle $E$ which is the extension of $\Q_1$ by $\End(TX)$, given by the short exact sequence
 \begin{equation}
  0\rightarrow \End(TX)\xrightarrow{\iota_2} E \xrightarrow{\pi_2} \Q_1 \rightarrow 0~,
  \label{eq:ses2}
  \end{equation}
  with extension class $\cal R^I$. 
 There is a holomorphic structure on $E$ defined by the exterior derivative $\bar\partial_2$ 
 \begin{equation*}
\bp_2\;=\; 
\left[ \begin{array}{ccc}
\bp_{\vartheta^I} & 0 &\; {\cal R}^I\\
0 & \;\bp_{\cal A} & \;{\cal F} \\
0  &  0  & \;\bp  \end{array} \right],
\end{equation*}
which acts on $\Omega^{(0,q)}(E)$ and squares to zero.  

\begin{Cor}\label{prop:four}
\[\bp_2^2 = 0~.\]
\end{Cor}

\begin{proof}
Let
\begin{equation*}
\left(
\begin{array}{c}
\kappa\\ \alpha\\ \Delta
\end{array}\right)
\in \Omega^{(0,q)}(X, {\rm End}(TX)) \oplus\Omega^{(0,q)}(X,{\rm End}(V)) \oplus \Omega^{(0,q)}(X, TX)~.
\end{equation*}
Then 
\begin{equation*}
\bp_2^2 
\left(
\begin{array}{c}
\kappa\\\alpha\\ \Delta
\end{array}\right) =
\left[ 
\begin{array}{c}
\bp_{\vartheta^I}{\cal R}^I\big(\Delta\big) +  {\cal R}^I\big(\bp\Delta\big) \\
\bp_{\cal A}{\cal F}\big(\Delta\big) +  {\cal F}\big(\bp\Delta\big) \\
0    \end{array} 
\right] = 0~,
\end{equation*}
by theorems \ref{prop:one} and \ref{prop:three}.
\end{proof}
\vskip10pt

The infinitesimal moduli space of the holomorphic structure $\bp_2$ on the extension bundle $E$, which is given by
\[{\cal TM}_2 = H^{(0,1)}_{\bp_2}(X, E)~,\]
can be computed by a long exact sequence in cohomology as in the previous section
\begin{align*}
0 &\rightarrow H^{(0,1)}(\End(TX)) \xrightarrow{\iota_2'} H^{(0,1)}(E) \xrightarrow{\pi_2'} H^{(0,1)}(Q_1) \\
&\xrightarrow{\cal R} H^{(0,2)}(\End(TX)) \rightarrow H^{(0,2)}(E) \rightarrow \ldots
\end{align*}
where the Atiyah map ${\cal R}^I$ is the connecting homomorphism as can be deduced from theorem \ref{prop:three}. Note that in this computation we have used equation \eqref{eq:nosecQ1}, and so the Atiyah map ${\cal R}^I$ acts trivially from the zeroth level to the first level. This induces a splitting between the zeroth and first level of the long exact sequence, and so
\begin{equation}
 H^0(X, E) \cong H^0(Q_1)\oplus H^0(X, \End(TX))\cong H^0(X, \End(V))~.\label{eq:nosecQ2} 
 \end{equation}
The last equality follows from the stability of $TX$, and the fact that the endomorphisms in $spin(6)$ are traceless\footnote{This is true as $X$ is orientible and we require the connection $\nabla^I$ to be metric.}.
Then we find that the infinitesimal moduli space of the extension $E$ is
\[{\cal TM}_2= H_{\bp_2}^{(0,1)}(X, E) =
H_{\bp_{\vartheta^I}}^{(0,1)}(X,{\rm End}(TX))\oplus
H_{\bp_{\cal A}}^{(0,1)}(X,\End(V))\oplus (\ker({\cal F})\cap\ker({\cal R}^I)) 
~.\]
We remark again that the deformations in $H_{\bp_{\vartheta^I}}^{(0,1)}(X,{\rm End}(TX))$ should not correspond to any physical fields.

\subsection{Stability and Variations of the Primitivity Conditions for the Curvatures.}
\label{subsec:Primitive}

Before considering the constraints from the anomaly cancelation condition, in this section we discuss variations of the primitivity conditions for the curvatures of $V$ and $TX$
\[ \omega\lrcorner F = 0 ~, \qquad \omega\lrcorner R^I = 0~.\]
These conditions should be preserved under a general deformation, in particular under the deformations of the bundle $E$ discussed earlier, but also including deformations of the hermitian parameters.  In fact,  a polystable bundle remains polystable~\citep{MR2665168} under deformations $\Delta$ of the complex structure $J$ of $X$ which preserve the holomorphicity of $V$ and $TX$, that is for
\[\Delta \in \ker({\cal F})\cap\ker({\cal R}^I)~.\]  
Moreover, the theorem of Li and Yau~\citep{MR915839} guarantees that as the deformed bundles $V_t$ and $(TX)_t$ are polystable and holomorphic, then there are connections on $V_t$ and $(TX)_t$ which satisfy the instanton equations, in particular, such that the deformed curvatures are primitive with respect to the hermitian structure.

We generalise this result below to include deformations of the hermitian structure so that, for $X$ with a heterotic $SU(3)$ structure, in particular on a conformally balanced manifold, a general variation of the primitivity conditions of the curvatures which preserves the primitivity conditions does not pose  any constraints on the first order moduli space whenever the bundle is stable. 

We study the gauge bundle. A completely analogous result is obtained for the instanton connection $\nabla^I$ on $TX$.
Under a general variation the instanton equation becomes
\[0 = \p_t(\omega\lrcorner F)
= \half\,\p_t\left(\omega^{mn}F_{mn}\right)
= \half\, \left((\p_t\omega^{mn})  F_{mn} + \omega^{mn}\p_t F_{mn}\right)
~, \]
and therefore
\begin{equation}
 \omega\lrcorner \p_t F = -  \half\, \p_t(\omega^{mn})  F_{mn} =   (h_t^{(1,1)})\lrcorner F ~.
  \label{eq:varinstone}
 \end{equation}
 This equation means that $F$ acquires a non-primitive part under a general deformation
 \[(\partial_t F)^{(1,1)} = \frac{1}{3}\, \big((h_t^{(1,1)})\lrcorner F\big)\, \omega + f_t~,\]
 where $f_t$ is a primitive $(1,1)$-form, $\omega\lrcorner f_t = 0$. 
 Note that this non-primitive part of $\p_t F$ depends on the variations of the hermitian form
 and it is needed so that $F_t$ is primitive with respect to $\omega_t$.
 
On the other hand, considering instead a general variation of $F$ using equation \eqref{eq:curvV}. We find
\begin{equation}
(\partial_t F)^{(1,1)} =  \bp_{\cal A} b_t + \p_{{\cal A}^\dagger} \alpha_t~,\label{eq:varinsttwo}
\end{equation}
where $- {\cal A}^\dagger$ is the $(1,0)$ part of the gauge connection $A$\footnote{We set $A = {\cal A} - {\cal A}^\dagger$ so that $A$ is antihermitian.}, $\alpha_t$ is the $(0,1)$ part of $\p_t A$ as before, $b_t$ is the $(1,0)$ part of $\p_t A$, and the operator  $\p_{{\cal A}^\dagger}$ is the $(1,0)$ part of the covariant exterior derivative $\d_A$ defined in equation \eqref{eq:dA}.  This operator is given by
\begin{equation} 
\partial_{{\cal A}^\dagger} \beta = \p\beta - [{\cal A}^\dagger, \beta] ~,
\label{eq:pAdag}
\end{equation}
where $\beta\in \Omega^{(0,q)}(X, \End(V))$. 
It is easy to prove that this operator also squares to zero, $\p_{{\cal A}^\dagger}{}^2 = 0$ only if $F^{(2,0)} = 0$. 

Putting together equations \eqref{eq:varinstone} and \eqref{eq:varinsttwo} we obtain a relation
\begin{equation*}
 (h_t^{(1,1)})\lrcorner F = \omega\lrcorner\left( \bp_{\cal A} b_t + \p_{{\cal A}^\dagger} \alpha_t\right)~,
\end{equation*}
which seems to represent a constraint on the moduli space of hermitian structures $h_t$.  However for stable bundles this is not the case.  

\begin{Theorem}\label{tm:zero}
Consider the relation 
\begin{equation}
(h_t^{(1,1)})\lrcorner F =  \omega\lrcorner\left( \bp_{\cal A} b_t + \p_{{\cal A}^\dagger} \alpha_t\right)
\label{eq:varF11}
\end{equation}
which gives the contribution to the non-primitive part in the $(1,1)$ variation of $F$
\[ (\p_t F)^{(1,1)} = \frac{1}{3}\, \big((h_t^{(1,1)})\lrcorner F\big)\, \omega + f_t~,\]
where $f_t$ is primitive with respect to $\omega$.
On a {\it conformally balanced manifold}, with a stable holomorphic vector bundle $V$, such that the endomorphisms of $V$ are traceless, there are no gauge bundle parameters on the right hand side of equation \eqref{eq:varF11}, and there is always a gauge transformation so that \eqref{eq:varF11} is satisfied for any variation $h_t$ of the hermitian structure $\omega$.
\end{Theorem}

\begin{proof}
Let $\hat g_{mn} = e^{- \phi}\, g_{mn}$ be the Gauduchon metric and $\widehat\omega = e^{-\phi}\, \omega$ be the corresponding Gauduchon hermitian form.  Let $\widehat\lrcorner$ and $\hat *$ be the contraction operator and the Hodge dual operator with respect to the Gauduchon metric respectively.  Then
\begin{equation*}
(h_t^{(1,1)})\lrcorner F 
= *\big((\bp_{\cal A} b_t + \p_{{\cal A}^\dagger} \alpha_t) \wedge *\omega\big)
= e^{2\phi}\, *\big((\bp_{\cal A} b_t + \p_{{\cal A}^\dagger} \alpha_t) \wedge\hat\rho \big)~,
\end{equation*}
where 
\begin{equation*}
\hat\rho = e^{-2\phi}\, \rho~,\qquad\rho = *\omega= \half\, \omega\wedge\omega~.
\end{equation*}
Because on a conformally balanced manifold $\d\hat\rho = 0$, we have
\begin{align*}
(h_t^{(1,1)})\lrcorner F 
&= e^{2\phi}\,*\big(\bp_{\cal A}( b_t \wedge\hat\rho) 
+ \p_{{\cal A}^\dagger} (\alpha_t \wedge \hat\rho)\big)
= e^{2\phi}\,* \big(\bp_{\cal A}\hat* (J(b_t)) + \p_{{\cal A}^\dagger} \hat * (J(\alpha_t ))\big)\\
&= i\, e^{-\phi}\,\hat* \big(\bp_{\cal A}\,\hat*\, b_t - \p_{{\cal A}^\dagger}\, \hat * \, \alpha_t \big)~,
\end{align*}
where we have used the fact that $b_t$ is a $(1,0)$-form and $a_t$ is a $(0,1)$-form.  We have also used
\[ \hat *\beta = e^{(p-3)\phi}\, *\beta~,\]
which is true for any $p$-form $\beta$ in six dimensions.
We now note that the operators on the right hand side in the last equality are the adjoints, with respect to $\hat g$, of the differential operators $\bp_{\cal A}$ and $\p_{{\cal A}^\dagger}$ given by
\begin{align*}
\bp_{\cal A}^\dagger &= - *\p_{{\cal A}^\dagger}\,*~,\\
\p_{{\cal A}^\dagger}^\dagger &= - *\bp_{\cal A}\ *~.
\end{align*}
Using these operators, we now have
\begin{equation*}
(h_t^{(1,1)})\lrcorner F = 
i\, e^{-\phi}\,\big(- \p_{{\cal A}^\dagger}^{\widehat\dagger} \,b_t   +  \bp_{\cal A}^{\widehat\dagger}\,\alpha_t \big)~,
\end{equation*}
where $\widehat\dagger$ means the adjoint of the operators taken with respect to the Gauduchon metric.
Consider now the Hodge decomposition of $\alpha_t$
\[ \alpha_t = \bp_{\cal A}\epsilon_t + \bp_{\cal A}^{\hat\dagger}\, \eta_t + \alpha_t^{har}~,\]
where $\epsilon\in \Omega^0(X, {\rm End} V)$, $\eta\in \Omega^{(0,2)}(X, {\rm End} V)$ and $\alpha_t^{har}$ is the $\bp_{\cal A}$-harmonic part of $\alpha_t$ (using the Gauduchon metric).  We have a similar decomposition for $b_t$ with respect to the operator $\p_{{\cal A}^\dagger}$,
\[ b_t = \p_{{\cal A}^\dagger}\tilde\epsilon_t + \p_{{\cal A}^\dagger}^{\hat\dagger}\, \tilde\eta_t + \tilde\alpha_t^{har}~.\]
Then
\begin{align*}
\bp_{\cal A}^{\widehat\dagger}\, \alpha_t  &=\bp_{\cal A}^{\widehat\dagger}\bp_{\cal A}\epsilon_t~,\\
\p_{{\cal A}^\dagger}^{\widehat\dagger} \,b_t &=  \p_{{\cal A}^\dagger}^{\widehat\dagger}\p_{{\cal A}^\dagger}\epsilon_t~.
\end{align*}
Then, equation \eqref{eq:varF11} becomes
\begin{equation}
i\, e^{\phi}\, (h_t^{(1,1)})\lrcorner F = 
-\, \bp_{\cal A}^{\widehat\dagger}\bp_{\cal A}\epsilon_t
+ \p_{{\cal A}^\dagger}^{\widehat\dagger}\p_{{\cal A}^\dagger}\tilde\epsilon_t
~.\label{eq:varF11prefinal}
\end{equation}
Any variations $a_t$  and $b_t$ of $A$ corresponding to a gauge transformation, and which is therefore trivial,
is of the form
\[ \alpha_t = \bp_{\cal A}\epsilon_t~,\qquad b_t = \p_{{\cal A}^\dagger}\tilde\epsilon_t~.\]
for some $\{\epsilon_t,\tilde\epsilon_t\}\in\Omega^0(X, {\rm End} V)$.  Therefore there are no bundle parameters on the right hand side of equation \eqref{eq:varF11prefinal}.  
Consider the Laplacians 
\[\Delta_{\bp_{\cal A}} = \bp_{\cal A}^\dagger\, \bp_{\cal A}
+ \bp_{\cal A}\, \bp_{\cal A}^\dagger~,\quad{\rm and}\quad
\Delta_{\p_{{\cal A}}^\dagger} = \p_{\cal A}^\dagger\, \p_{\cal A}
+ \p_{\cal A}\, \p_{\cal A}^\dagger~,\]
and let $\widehat\Delta_{\bp_{\cal A}}$ and $\widehat\Delta_{\p_{{\cal A}}^\dagger}$ be the corresponding Laplacians with respect to the Gauduchon metric.
Then we can write equation \eqref{eq:varF11prefinal} as
\begin{equation}
e^{\phi}\, (h_t^{(1,1)})\lrcorner F = \hat h_t^{(1,1)}\,\widehat\lrcorner F = 
i\, \widehat\Delta_{\bp_{\cal A}}\epsilon_t -i\widehat\Delta_{\p_{{\cal A}}^\dagger}\tilde\epsilon_t~,
\label{eq:varF11final}
\end{equation}
where $\hat h_t = e^{-\phi}\, h_t$ and $\widehat\lrcorner$ is the contraction operator with respect to the Gauduchon metric.

This equation means that $\hat h_t^{(1,1)}\,\widehat\lrcorner F$, which belongs to the space $ \Omega^0(X, {\rm End}(V)$, is in the image of Laplacians which are elliptic operators.  Therefore, whenever the kernel of these Laplacians is trivial, the image of the Laplacians spans all of the space $ \Omega^0(X, {\rm End}(V)$ and equation \eqref{eq:varF11final} always a solution for any  $\hat h_t^{(1,1)}$.  This is precisely the case for a {\it stable} bundle $V$ because $H^0(\End V) = 0$ for traceless endomorphisms.

We conclude that, for stable vector bundles $V$ with traceless endomorphisms, equation \eqref{eq:varF11final} poses no constraints on the deformations of the hermitian moduli. 

\end{proof}
\vskip10pt
A similar result follows for the tangent space $TX$. We see then that variations of the instanton equations 
\[ \omega\lrcorner F = 0~, \quad{ \rm and} \quad \omega\lrcorner R^I = 0~,\]
impose no constraints on the variations $h_t$ of the hermitian form $\omega$, nor do they give  a relation between these and the moduli of the bundles, provided the bundles are {\it stable} with traceless endomorphisms.   It should be noted however that first order deformations may be obstructed and that stability or the Yang-Mills conditions may be spoiled.

If, on the other hand, the bundle $V=\oplus_{i=1} V_i$ is {\it polystable}, we then need to satisfy the Yang-Mills condition for each separate factor $V_i$, each of which need not be separately traceless and thus could have non-trivial zeroth cohomology. In this case, \eqref{eq:varF11final} could potentially constrain the hermitian moduli for each bundle factor $V_i$ of non-trivial trace. \footnote{We would like to thank James Gray for pointing this out. The first version of this paper on the arXiv did not include this subtlety.}${}^,$\footnote{See also ~\citep{Anderson:2009nt}, where the Yang-Mills conditions where related to $D$-term conditions in the-four dimensional effective field theory, in the case of polystable sums of line bundles on Calabi-Yau manifolds.}

We will come back to these issues, and in particular to the general case of  polystable bundles, when we include the anomaly cancelation conditions in the context of the Strominger system.  As we will see the constraints in equation \eqref{eq:varF11final} are naturally taken care of in our computations of moduli space of the Strominger system.

\subsection{Constraints from the Anomaly Cancellation Condition.}
\label{subsec:Anomaly}

We construct an extension bundle $\Q$ of $E$ such that $\Q$ has a holomorphic structure, and which allows for the implementation of the anomaly cancelation condition equation \eqref{eq:BIanomaly}
\begin{equation*}
\d H = -2 i \partial\bar\partial\omega = \frac{\alpha'}{4}\,\left(\tr (F\wedge F) - \tr (R^I\wedge R^I) \right).
\end{equation*}
Moreover, using deformation theory of holomorphic bundles, we will show that this construction results in a description of the moduli space of heterotic compactifications $(X,V)$.  

We begin by defining a map $\cal H$ as follows:
\begin{equation}
{\cal H}\;:\;\Omega^{(0,q)}(X, E)\longrightarrow\Omega^{(0,q+1)}(X, T^*{}^{(1,0)}X),\label{eq:mapH}
\end{equation}
by
\begin{equation}
{\cal H}(x)_m= i\, (-1)^q\, \Delta^p\wedge Q_{n}{}^r\,(\p\omega)_{pmr}\, \d x^n-\frac{\a}{4}(\tr\,(f_m\wedge \alpha)-\tr\,(r_m^I\wedge \kappa))~,\label{eq:Hdef}
\end{equation}
where
\[ f_m = F_{mq}\,Q_{n}{}^q\,\d x^n~, \qquad\text{and}\quad r_m^I = R^I_{mq}\, Q_{n}{}^q\,\d x^n~,\]
and 
\begin{equation*}
x=\left(
\begin{array}{c}
\kappa\\ \alpha\\ \Delta
\end{array}
\right)
\in\Omega^{(0,q)}(X,E)~.
\end{equation*} 
and where, as before, $\Delta$ is valued in $T^{(1,0)}X$,  $\alpha$ is valued in $\End(V)$, and $\kappa$ is valued in $\End(TX)$.

\begin{Theorem}\label{prop:five}
\begin{equation}
\bp({\cal H}(x)) = - {\cal H}(\bp_2(x))~,\quad\forall x\in \Omega^{(0,q)}(X,E)~, 
\label{eq:Hcohom} 
\end{equation}
and therefore the map $\cal H$ is a map between cohomologies
\begin{equation}
{\cal H}\;:\;H_{\bp_2}^{(0,q)}(X, E)\longrightarrow H_{\bp}^{(0,q+1)}(X, T^*{}^{(1,0)}X).\label{eq:mapHcohom}
\end{equation}
\end{Theorem}

\begin{proof}
Recall
\begin{equation*}
\bp_2\, x = \left(
\begin{array}{c}
\bp_{\vartheta^I}\kappa + {\cal R}^I(\Delta)\\
\bp_{\cal A}\alpha + {\cal F}(\Delta)\\
\bp\Delta
\end{array}\right)~,
\end{equation*}
Then
\begin{equation}
\begin{split}
{\cal H}(\bp_2\, x)_m &= - i\, (-1)^q\, \bp\Delta^p\wedge Q_n{}^r\, (\p\omega)_{pmr}\,\d x^n\\
&  - \frac{\alpha'}{4}\, \left(
  {\rm tr}\left(f_m\wedge(\bp_{\cal A}\alpha + {\cal F}(\Delta))\right)
- {\rm tr}\left(r_m^I\wedge(\bp_{\vartheta^I}\kappa + {\cal R}^I(\Delta))\right)
\right)~,\label{eq:Hd2x}
\end{split}
\end{equation}
and we obtain
\begin{equation}
\begin{split} 
\bp({\cal H}(x))_m + {\cal H}(\bp_2\, x)_m &=
 i\, \Delta^p\wedge\bp\big( Q_n{}^q\, (\p\omega)_{pmq}\,\d x^n\big)\\
& - \frac{\alpha'}{4}\, \Delta^p\wedge \left(
   {\rm tr}\big(f_m\wedge f_p\big) 
   - {\rm tr}\big(r^I_m\wedge r^I_p\big)
\right) \\[3pt]
&  - \frac{\alpha'}{4}\,\left(
{\rm tr}\big( \bp_{\cal A}f_m\wedge\alpha\big) 
-{\rm tr}\big( \bp_{\vartheta^I}\,r^I_m\wedge\kappa\big) 
  \right)~. \label{eq:theo}
\end{split}
\end{equation}
The last two terms vanish because of the Bianchi identities for $F$ and $R^I$, in particular, because of equation  \eqref{eq:strongBIF}
\begin{equation*}
P_{m}{}^p\,\bp_{\cal A} \big(f_p\big)~ = 0,  
\end{equation*}
and the analogous one for $R^I$. The other terms cancel due to the Bianchi identity of the anomaly cancelation condition \eqref{eq:BIanomaly}.
In fact, the Bianchi identity is equivalent to
\[ 4 i \, Q_{[m}{}^r\, \partial_{|r|}\big(P_n{}^s\p_{|s|}\omega_{pq]}\big)
= \frac{\alpha'}{4}\, \big( \tr(F_{[mn}F_{pq]}) - \tr(R^I_{[mn}R^I_{pq]}) \big)~,\]
which implies
\[ \frac{\alpha'}{4}\, \big( \tr(f_m\wedge f_p) - \tr (r^I_m\wedge r^I_p)\big)
=
i\, P_m{}^r \, P_p{}^q\, \bp\big( (\p\omega)_{qrn}\, \d x^n\big)~.\]
Therefore
\[ \bp({\cal H}(x))_m + {\cal H}(\bp_2\, x)_m = 0~.\]
\end{proof}

\vskip5pt
The Atiyah map $\cal H$ is well defined as a map between cohomologies.  To see this we need to prove that the class ${\cal H}(x)\in H_{\bp}^{(0,q+1)}(X, T^*{}^{(1,0)}X)$ and that equation \eqref{eq:Hcohom} are invariant under gauge transformations.
Recall that under a gauge transformation
\[ {\cal A} \mapsto \Phi({\cal A} - \Phi^{-1}\bp \Phi)\Phi^{-1}~,\]
where $\Phi$ takes values in the Lie algebra of the structure group of the bundle $V$.  This implies that
\[ \alpha_t \mapsto  \Phi(\alpha_t - \bp_{\cal A}(\Phi^{-1}\p_t \Phi))\Phi^{-1}~.\]
Let $\alpha\in\Omega^{(0,q)}(X, {\rm End}(V))$.  Then, under a gauge transformation
\[ \alpha \mapsto \Phi(\alpha - \bp_{\cal A}Y)\Phi^{-1}~,\]
where $Y\in\Omega^{(0,q-1)}(X, {\rm End}(V))$.  Thus, the term $\tr(f_m\wedge\alpha)$ in ${\cal H}(X)$ transforms as
\begin{equation*}
\begin{split}
 \tr(f_m\wedge\alpha)&\mapsto
 \tr(\Phi f_m\Phi^{-1}\wedge\Phi(\alpha - \bp_{\cal A}Y)\Phi^{-1}) \\ 
&\mapsto \tr(f_m\wedge\alpha) + \bp(\tr(f_m\wedge Y)) - \tr(\bp_{\cal A}f_m\wedge Y)~.
 \end{split}
 \end{equation*}
As the last term vanishes due to the Bianchi identity for $F$, we find that under a gauge transformation ${\cal H}(x)$ 
changes only by a $\bp$-exact part, and therefore the class ${\cal H}(x)\in H_{\bp}^{(0,q+1)}(X, T^*{}^{(1,0)}X)$ is gauge invariant. 
To prove that equation \eqref{eq:Hcohom} is gauge invariant, we note first that $\bp{\cal H}(x)$ is invariant.  
On the other hand, ${\cal H}(\bp_2 x)$ is also invariant
because 
\[ \bp_{\cal A}\alpha\mapsto \Phi(\bp_{\cal A}\alpha)\Phi^{-1}~,\]
and so the term $\tr(f_m\wedge(\bp_{\cal A}\alpha + {\cal F}(\Delta)))$ in ${\cal H}(\bp_2 x)$ is invariant (see equation \eqref{eq:Hd2x}).  The argument for the other term $\tr(r^I_m\wedge(\bp_{\vartheta^I}\kappa + {\cal R}(\Delta)))$ in ${\cal H}(\bp_2 x)$ is similar.

We now construct a bundle $\Q$ by extending $E$ by  $T^*X$ given by the short exact extension sequence
\footnote{This structure is similar to the one which appeared in \citep{2013arXiv1304.4294G, 2013arXiv1308.5159B} in the context of generalised geometry for heterotic compactifications.  It would be interesting to find out the precise relation.}
\begin{equation}
0\rightarrow T^*X\xrightarrow{i}\Q \xrightarrow{\pi}E\rightarrow 0\:,
\end{equation}
with extension class $\cal H$.
We define a holomorphic structure on $\Q$ by defining the operator $\bar D$ on $\Q$ 
\begin{equation}
\bar D\;=\; \left[ \begin{array}{cc}
\bp & \,{\cal H} \\
0  & \,\bp_2  \end{array} \right].\label{eq:barD}
\end{equation}
Clearly, by theorems \ref{prop:four} and  \ref{prop:five}
\[ \bar D^2 = 0~.\] 

It is worth pointing out that the construction of the operator $\bar D$ such that it squares to zero, 
a condition we have seen is equivalent to 
\begin{equation*}
{\cal H}(\bp_2 x)_m=-\bp{\cal H}(x)_m\;\;\;\forall\;\;\;x\in\Omega^{(0,q)}(X, E)\:,
\end{equation*}
implies the Bianchi identity \eqref{eq:anomaly}.  This is clear as this equation implies (see equation \eqref{eq:theo})
\begin{equation*}
\Delta^p\wedge{\cal H}_{pm}=0\;\;\;\forall\;\;\;\Delta\:,
\end{equation*}
where 
\begin{equation*}
{\cal H}_{pm}=i\, \bp\big((\p\omega)_{pmq}\,Q_n{}^q\, \d x^n\big)
- \frac{\alpha'}{4}\, \big( \tr(f_m\wedge f_p) - \tr (r^I_m\wedge r^I_p)\big)~.
\end{equation*}
In particular, 
\begin{equation*}
g^p\, {\cal H}_{pm} =0\;\;\;\forall\;\;\;g\in\Omega^0(X,TX).
\end{equation*}
It follows that ${\cal H}_{mn}=0$, which is equivalent to the Bianchi Identity. We thus have that $\bar D^2=0$ if and only if the Bianchi identities for $F$, $R^I$ and $H$ are satisfied.

Deformations of the holomorphic structure determined by $\bar D$ correspond to elements of $H^{(0,1)}_{\bar D}(X, \Q)$. We will compute this cohomology by the usual means of a long exact sequence in cohomology. We have defined above a short exact extension sequence
\begin{equation}
0\rightarrow T^*X\xrightarrow{\iota}\Q\xrightarrow{\pi}E\rightarrow0\:,\label{eq:sesH}
\end{equation}
with extension class $\cal H$. This gives rise to a long exact sequence in cohomology
\begin{equation} 
\begin{split}
& 0 \rightarrow H^0(T^*X)\xrightarrow{\iota'} H^0(\Q) \xrightarrow{\pi'} H^0(E) \\
&\xrightarrow{{\cal H}_0} H^1(T^*X) \xrightarrow{\iota'} H^1(\Q) \xrightarrow{\pi'} H^1(E) \\
&\xrightarrow{{\cal H}_1} H^{2}(T^*X) \rightarrow H^2(\Q) \rightarrow \ldots
\end{split}
\label{eq:les3}
\end{equation}
where, by theorem \ref{prop:five}, the connecting homomorphism  is $\cal H$, and where we denote by ${\cal H}_q$ the map $\cal H$ when we need to make it clear that it is acting on $(0,q)$-forms.
In the long exact sequence above, note that
\[H^0(X,T^*X)=0~,\]
where the vanishing of this cohomology follows from the fact that
\[ H_{\bp}^{(0,3)}(X, TX) = 0~,\]
by zero-slope stability and
\[ H_{\bp}^0(X,T^*X) \cong
H_{\bp}^{(1,0)}(X) \cong
H_{\p}^{(0,1)}(X) \cong
H_{\bp}^{(2,3)}(X) \cong
H_{\bp}^{(0,3)}(X, TX) ~.\]
The second isomorphism is due to complex conjugation and the third comes from Hodge duality.  The fourth isomorphism is given by the holomorphic no-where vanishing $(3,0)$ form $\Omega$~\cite{MR915841}. For every element  $\beta^m\in H_{\bp}^{(0,q)}(X, TX)$, there is  an element 
\[ \Omega(\beta) = \frac{1}{2!\,q!}\, \beta^m\wedge\Omega_{mnp}\, \d x^n\wedge\d x^p\ \in \, H_{\bp}^{(2,q)}(X)~.\]
The map is an isomorphism because of the properties of $\Omega $ and the fact that 
\[ \Omega(\bp\beta) = \bp(\Omega(\beta))~.\]

We are now ready to write the infinitesimal moduli space of holomorphic structures of the extension $\Q$. By exactness of the sequence \eqref{eq:les3}, it follows that
\begin{equation}
\label{eq:defBI}
H_{\bar D}^1(X, \Q)\cong\Im(i')\oplus\Im(\pi')\cong \Big[H_{\bp}^1(X,T^*X)\Big/\textrm{Im}({\cal H}_0)\Big]\oplus\ker({\cal H}_1)\:,
\end{equation}
is the tangent space to the moduli space of deformations of the holomorphic structure defined by $\bar D$  on $\cal Q$.
As we have remarked, the Bianchi identities give rise to a holomorphic structure on $\Q$ defined by $\bar D$ and $\cal H$. The elements in the factor 
\[\ker({\cal H}_1)\subseteq H_{\bp_2}^{(0,1)}(X, E)~,\]
 correspond to those deformations of the holomorphic structure on $E$ which preserve the holomorphic structure of the co-tangent bundle $T^*X$. and the elements in the factor 
\[ {\cal M}_{HS} = \Big[H_{\bp}^1(X,T^*X)\Big/\textrm{Im}({\cal H}_0)\Big]\]
are the moduli of the (complexified) hermitian structure. 
In the following subsections we interpret in detail the elements in $H_{\bar D}^{(0,1)}(X,{\cal Q})$,  which by construction should be precisely the infinitesimal moduli space of the Strominger system.

\subsection{The Yang-Mills Condition Revisited.}
In the computation leading to  \eqref{eq:defBI}, we found that we need to take the quotient by 
\begin{equation*}
\textrm{Im}({\cal H}_0)\cong\{\tr({\cal H}_0({x}))\:\vert\:\ x\in H^0(X,E)\}\:.
\end{equation*}
Noting that (see equation \eqref{eq:mapH})
\[ {\cal H}_m(x)\wedge\d x^m = \frac{\alpha'}{4}\, \tr(F\,\alpha)~,\quad \alpha\in H^0(X,\End(V))\]
we find that
\[ \textrm{Im}({\cal H}_0)\cong\{\tr(F\, \alpha)\:\vert\:\ \alpha\in H^0(X,\End(V))\}\subset H^{(1,1)}(X)\:.\]
which may be non-trivial whenever $H^0(X,\End(V))$ is non-trivial, that is, when the bundle $V=\oplus_iV_i$ is polystable with bundle factors $V_i$ for the which ${\rm End}(V_i)$ has non-vanishing traces. 
Let $V_i$ be such a stable bundle with ${\rm End}(V_i)$ has non-vanishing traces, and let
\begin{equation*}
\alpha_i\in H^0(X,\End(V_i))=\mathbb{C}\:,
\end{equation*}
where the $\mathbb{C}$ corresponds to the trace of the endomorphisms. These correspond of course to sections of $\End(V_i)$ by the Dolbeault theorem. Without loss of generality, we may assume that this section takes the form $c_i I_i$, where $c_i$ is a constant, and $I_i$ is the identity on isomorphisms, which is part of the Lie-algebra for algebras of non-trivial trace. We may therefore assume that a generic section takes the form
\begin{equation}
\alpha=\sum_i c_i I_i\:,\label{eq:gensec}
\end{equation}
where the constants  $c_i$ are such that $\alpha$ is traceless. It follows that the elements in $\textrm{Im}({\cal H}_0)$ are of the form\footnote{
Note that ${\rm Im}({\cal H}_0)=\{\sum_ic_i\tr F_i\}$ without any further constraints on the constants $c_i$. This is due to the fact that 
$\sum_ic_i\tr F_i = \sum_i(c_i+K)\tr F_i$  for any constant $K$, as  $\sum_i\tr F_i = 0$.}

\begin{equation*}
\label{eq:Dtermcond}
[h]=\sum_i c_i[\tr(F_i)]\:,
\end{equation*}
where the brackets refer to cohomology classes. 

We claim that this is precisely the constraint on the moduli enforced by the Yang-Mills condition. As we have seen (Theorem \ref{tm:zero}), the Yang-Mills conditions pose no extra conditions on the moduli for {\it stable} bundles. If, on the other hand, the vector bundle is {\it polystable}, then these conditions may introduce constraints on the moduli. The constraint is exactly of the form above, and we take a moment to explain why.

Let $V_i$ be a stable bundle of nonzero trace. As $V=\oplus_i V_i$ is polystable 
\[\mu(V_i) = \mu(V) = 0~, \]
we must have that the Yang-Mills condition for a bundle $V_i$ is,
\begin{equation*}
\omega\lrcorner F_i=0\:.
\end{equation*}
As noted before, it is only the trace part of the bundle that can impose non-trivial constraints from this condition. Taking the trace and using instead the Gauduchon metric $\hat\omega$ this condition becomes
\begin{equation}
\label{eq:YMi}
\hat\omega\hat\lrcorner\,\tr\:F_i=0\:.
\end{equation} 
Varying equation \eqref{eq:YMi}, and performing a computation similar to that leading to equation \eqref{eq:varF11final}, we obtain that on a {\it conformally balanced manifold}
\begin{equation*}
\p_t\hat\omega\,\hat\lrcorner\,\tr F_i\:\in\:\textrm{Im}(\widehat\Delta_\p)+\textrm{Im}(\widehat\Delta_{\bp})\:.
\end{equation*}
Equivalently, this condition means that
\begin{equation*}
(\p_t\hat\omega,\tr\:F_i) = 0\:,
\end{equation*}
where the integration is done with respect to the Gauduchon metric. Considering the Hodge decomposition of $\p_t\hat\omega$ with respect to the $\bp$ operator and the Gauduchon metric, it is easy to see that the $\bp^{\hat\dagger}$-exact piece  drops out from the inner product.  Hence, only the $\bp$-closed part contributes that is,  the elements in $H_{\bp}^{(1,1)}(X)$\footnote{By Proposition \ref{prop:CBcond}, the $\bp$-exact part is determined entirely by deformations of the complex structure.}.  These correspond to the (imaginary part of) the hermitian moduli.  However, the vanishing of the inner product implies that we should also mod out by forms proportional to $\tr\:F_i$ in the hermitian moduli, or more generally, by terms proportional to $\sum_ic_i\tr\:F_i$. 

Interestingly, by computing the first cohomology $H^1_{\bar D}(X,\Q)$, which gives the tangent space $\cal TM$ of the moduli space of holomorphic structures on $\Q$ at $\bar D$, we find that the instanton condition gets implemented for free. This is not surprising, as discussed in the next section, where we consider $\cal TM$  in more detail. As we will see, this is naturally included in the quotient by $\bar D$-exact terms.

\subsection{The Moduli Space of the Strominger System.}

We now claim that the tangent space of the moduli space of the Strominger system, is given by $H_{\bar D}^{(0,1)}(X,{\cal Q})$ in equation \eqref{eq:defBI}.   The extension bundle $\cal Q$, with extension class $\cal H$ in equations \eqref{eq:sesH} and \eqref{eq:Hdef},
with the holomorphic structure $\bar D$ in equation \eqref{eq:barD} determined by the Bianchi identities, together with the requirement that the bundles $V$ and $TX$ are polystable  and that $X$ is conformally balanced is equivalent to the Strominger system.  
The holomorphic structure $\bar D$ includes the requirement that $V$ and $TX$ should be holomorphic and that $X$ must have an integrable complex structure. Moreover, as we have seen it also implements the anomaly cancelation condition  together with the fact that $J$ and $\omega$ are covariantly constant with respect to the the Bismut connection (recall that this is reflected in the fact $ H = J(\d\omega)$).  Also, the instanton conditions are satisfied automatically by the Theorem of Li and Yau.  It is natural to expect therefore that deformations  of this structure gives variations of the Strominger system, except that we have to take care of the conformally balanced condition.  In this section we elaborate on these issues. 

Consider the elements in the cohomology
\begin{equation}
\label{eq:defMSS}
H_{\bar D}^1(X, \Q)\cong {\cal M}_{HS}\oplus\ker({\cal H}_1)\:,
\qquad  {\cal M}_{HS} = \Big[H_{\bp}^1(X,T^*X)\Big/\textrm{Im}({\cal H}_0)\Big]~,
\end{equation}
which we would like to interpret as the moduli of the Strominger system. 
The cohomology group  $H_{\bar D}^1(X, \Q)$ is of course the tangent space to the moduli space of deformations of the holomorphic structure on
$\cal Q$ given by the differential operator $\bar D$ in equations \eqref{eq:barD} and \eqref{eq:Hdef}.  The key issue here is that by preserving the holomorphic structure on $\cal Q$ these moduli correspond to deformations which preserve the Bianchi identities.

We begin with the $\bar D$-closed elements
\begin{equation}
 {\cal H}_1(x_t)_m = - \bp y_{t\, m}~,\qquad \bp_2 x_t = 0~,\label{eq:Dclosed}
\end{equation} 
for $x_t\in \Omega^{(0,1)}(X,E)$ and $y_t\in \Omega^{(0,1)}(X,T^*X)$. Clearly, the  left hand side of the first equation only involves $x_t\in H_{\bp_2}^{(0,1)}(X,E)$, that is, only involves variations of the holomorphic structure of $E$. Hence, the moduli in
\[ \ker{\cal H}_1\subseteq H_{\bp_2}^{(0,1)}(X,E)~,\]
represent those deformations of the holomorphic structure of $E$ which preserve the holomorphic structure on the cotangent bundle $T^*X$.
In preserving these holomorphic structures the Bianchi identities are therefore preserved.  One can also see this explicity (see below).
On the other hand, for a fixed holomorphic structure on $E$, that is for $x_t = 0$, we have that $\bp y_t = 0$ and so 
the moduli in 
\[{\cal M}_{HS}= \Big[H_{\bp}^1(X,T^*X)\Big/\textrm{Im}({\cal H}_0)\Big]~\]
correspond to the (complexified) cotangent bundle moduli.

Consider now the $\bar D$-exact forms. Let
\begin{equation*}
\left( \begin{array}{c}
y_t \\ x_t
\end{array}\right)
\in \Omega^{(0,1)}(X, {\cal Q})~,\qquad
x_t = 
\left( \begin{array}{c}
\kappa_t \\ \alpha_t \\ \Delta_t
\end{array}\right)
\in \Omega^{(0,1)}(X, E)~,
\end{equation*}
and 
\begin{equation*}
\left( \begin{array}{c}
f_t \\ \xi_t
\end{array}\right)
\in \Omega^0(X, {\cal Q})~,\qquad
\xi_t  =  
\left( \begin{array}{c}
\eta_t \\ \epsilon_t \\ \delta_t
\end{array}\right)
\in \Omega^0(X, E)~.
\end{equation*}
The $\bar D$-exact forms satisfy
\begin{equation}
\left( \begin{array}{c}
 y_t \\ x_t
\end{array}\right)
=
\left( \begin{array}{cc}
\bp f_t + {\cal H}_0(\xi_t) \\
\bp_2  \xi_t \end{array} \right).\label{eq:Dexact}
\end{equation}
The second equation are the trivial deformations of the holomorphic structure on $E$ corresponding to changes in $J$ due to diffeormophisms
\[\Delta_t = \bp \delta_t~,\]  
changes of the gauge fields from gauge transformations and trivial deformations of $J$
\[  \alpha_t = \bp_{\cal A}\epsilon_t + {\cal F}(\delta_t)~,\]
and a similar equation for the trivial deformations of the tangent bundle
\[\kappa_t = \bp_{\vartheta^I}\eta_t + {\cal R^I}(\delta_t)~.\]
The first equation in \eqref{eq:Dexact} can be written as
\begin{equation}
y_t = \bp f_t + {\cal H}_0(\xi_t) = \bp f_t -\frac{i}{2}\, \delta_t^p \, (\partial \omega)_{pmn}\, \d x^m\wedge \d x^n + 
\frac{\alpha'}{4} \Big( \tr(\epsilon_t F) - \tr(\eta_t R^I)\Big)~. \label{eq:trivialy}
\end{equation}
The last three terms come from trivial deformations of the holomorphic structure of $E$.  Keeping fixed the deformations of the holomorphic structure on $E$, that is, setting 
\[\bp_2\xi_t = 0~,\]
we see that the last term vanishes due to the stability of the tangent bundle $TX$, which implies that for traceless endomorphisms of $TX$ there are no sections with values in $TX$
\[ H^0({\rm End} (TX)) = 0~.\]
In this case, the second term, which corresponds to trivial deformations of the complex structure due to diffeomorphisms of $X$, also vanishes as there are no sections of with values in $TX$. 
The third term corresponds to the discussion in the previous section.  In fact, since a generic section of ${\rm End}(V)$ takes the form in equation \eqref{eq:gensec} we have that  this term is of the form
\[ \tr(\epsilon_t F) = \sum_i\,\tr( c_i F_i) = [h]~, \]
where $[h]$ represents a class in $H^{(1,1)}(X)$.  As we argued in the previous section this implements the instanton condition on the polystable bundle $V$.

We still need to discuss the meaning of the  first term  in equation \eqref{eq:trivialy}, which is related to the preservation of the conformally balanced condition.
We claim that  the variations in $H_{\bar D}^{(0,1)}(X, {\cal Q})$ preserve the conformally balanced condition.  Our results above imply that the deformations of the hermitian structure $y_t$ which {\it preserve the anomaly cancelation condition}, are  $(1,1)$ forms which are $\bp$-closed, and that the $\bp$-exact part is trivial.  The fact that the $\bp$-exact part of $y_t$ is trivial is {\it precisely} the content of Proposition \ref{prop:CBcond}.  In fact, in
Proposition \ref{prop:CBcond} it was proven that, as long as $TX$ is stable, the preservation of the conformally balanced condition ($\d\p_t\hat\rho = 0$) determines the $\bp$-exact part of the $\bp$-Hodge decomposition of the $(1,1)$ form $\hat*(\p_t\hat\rho)^{(2,2)}$
in terms of the deformations of the complex structure of $X$.

Finally, we would like to compare our results with those obtained by directly varying the anomaly cancelation condition.  Recall that
\begin{equation}
H =  i\, (\p- \bp)\, \omega = J(d\omega) = \d B + {\cal C S}~.\label{eq:anomalyA}
\end{equation}
where  
\[ {\cal C S} = \frac{\alpha'}{4}\, ({\rm CS}[A] - {\rm CS}[\Theta^I])~,\]
and ${\rm CS}[A]$ and ${\rm CS}[\Theta^I]$ are the Chern--Simons 3-forms for these connections defined by
\begin{equation*} 
{\rm CS}[A] = \tr\left(A\wedge\d A + \frac{2}{3} A\wedge A\wedge A\right)~, 
\end{equation*}
and similarly for ${\rm CS}[\Theta^I]$.  The Bianchi identity for the anomaly cancelation condition is 
\begin{equation*}
\d H = 2 i\, \bp \p \omega = \frac{\alpha'}{4}\,\left(\tr (F\wedge F) - \tr (R^I\wedge R^I) \right)~.
\end{equation*}
The variations of equation \eqref{eq:anomalyA} are given by~\citep{Anderson:2014xha, Candelas2014}
\begin{equation}
\begin{split}
\p_t H &= J(\d(\p_t\omega)) + (\Delta_t + \Delta^*_t)^p\wedge H_{pmn}\, \d x^m\wedge\d x^n\\
&= \frac{\alpha'}{2}\, \left(\tr (\p_t A\wedge F) - \tr(\p_t\Theta^I\wedge R^I)\right) + \d {\cal B}_t~,\label{eq:varanomaly}
\end{split}
\end{equation}
where
\begin{equation}
{\cal B}_t = \p_t B - \frac{\alpha'}{4}\, \left(\tr (A\wedge \p_t A) - \tr(\Theta^I\wedge\p_t\Theta^I)\right)~,\label{eq:varB}
\end{equation}
and
\[ \Delta_t = (\p_t z^a) \Delta_a~,\qquad \Delta^*_t = (\p_t \bar z^{\bar a})\, \bar\Delta_{\bar a}~,\]
with $\bar\Delta_{\bar a}$ the complex conjugate of $\Delta_a$.
Let
\begin{equation}
{\cal Z}_t = {\cal B}_t + i \p_t\omega~.\label{eq:complexherm}
\end{equation}
Separating equation \eqref{eq:varanomaly} by type we find
\begin{equation}
\begin{split}
(0,3)\, {\rm part:}\qquad &\bp{\cal Z}_t^{(0,2)} = 0\\
(1,2)\, {\rm part:}\qquad &\p{\cal Z}_t^{(0,2)} +  \bp{\cal Z}_t^{(1,1)}  = 2\, {\cal H}(x_t)_m\wedge\d x^m
\label{eq:eqsforZ}
\end{split}
\end{equation}
where
\[ 2\, {\cal H}(x_t)_m\wedge\d x^m = i \Delta_t{}^m \wedge (\p\omega)_{mnp}\, \d x^n\wedge \d x^p
- \frac{\alpha'}{2}\, \big(\tr(\alpha_t\wedge F) - \tr(\kappa_t\wedge R^I)\big)~,\]
for a deformation
\begin{equation*}
x_t=\left(
\begin{array}{c}
\kappa_t \\ \alpha_t\\ \Delta_t
\end{array}
\right)
\in H_{\bp_2}^{(0,1)}(X,E)
\end{equation*} 
of the holomorphic structure of $E$. Note how equation \eqref{eq:Dclosed} is more restrictive than equation \eqref{eq:eqsforZ}.   The reason for this extra constraint is that we have imposed a holomorphic structure on $T^*X$.  This means that the representative of the class
${\cal Z}^{(0,2)}\in H_{\bp}^{(0,2)}(X)$  must be such that
$\p{\cal Z}^{(0,2)}$ is $\bp$-exact if the deformed structure is to remain a holomorphic structure on $\cal Q$. 

A mild assumption on the cohomology of $X$ would guarantee that this condition is satisfied.  Suppose that
\begin{equation}
 H_{\bp}^{(0,1)}(X) = 0~.\label{eq:novectors}
 \end{equation}
This condition is very  interesting regarding deformations of the heterotic $SU(3)$ structure of the manifold $X$. 
It is not too hard to prove that this  is enough to guarantee that
\begin{equation*}
H_{\bp}^{(2,1)}(X)=H^{(2,1)}_{\d}(X)\:,
\end{equation*}
so that the allowed complex structure variations in this case are counted by the dimension of $H^{(0,1)}_{\bp}(X,TX)$, and not a subset of this (see section \ref{sec:varsJ} on deformations of the complex structure of $J$).  These matters are discussed further in~\citep{DKS2014}.

We claim that when \eqref{eq:novectors} is satisfied
\[ H_{\bp}^{(0,2)}(X) \cong H_{\bp}^{(0,1)}(X) = 0~.\]
To prove this, let $\beta$ be a $(0,2)$-form.  We can construct a $(1,0)$-form using $\Omega$ as
\begin{equation} 
\alpha = \beta\lrcorner\Omega = - *(\beta\wedge *\Omega) = i *(\beta\wedge \Omega)~,\label{eq:alphabeta}
\end{equation}
where we have used the fact that $*\Omega = -i \Omega$.  Conversely, we can construct a $(0,2)$-form $\beta$ given a $(1,0)$-form $\alpha$
\[ \beta = ||\Omega||^{-2}\, \alpha\lrcorner\bar\Omega~.\]
Then we find that
\[\p^\dagger\alpha =  i * \bp (\beta\wedge \Omega) = i *(\bp\beta\wedge\Omega)~,\]
and therefore
\[ \bp\beta = 0\qquad\iff\qquad\p^\dagger\alpha = 0~.\] 
Suppose now that 
\[\beta = \bp\lambda~,\]
 for some $(0,1)$-form $\lambda$.  Then equation \eqref{eq:alphabeta} gives 
 \begin{equation*}
 \alpha = (\bp\lambda)\lrcorner\Omega = i *((\bp\lambda)\wedge\Omega) = i * \bp(\lambda\wedge\Omega)
 =  - *\bp **(\lambda\wedge *\Omega) = \p^\dagger (\lambda\lrcorner\Omega)~,
 \end{equation*}
so if $\beta$ is $\bp$-exact, then $\alpha$ is $\p$-coexact.  Conversely, if 
\[ \alpha = \p^\dagger\gamma~,\]
for some $(2,0)$-form $\gamma$, then $\beta$ is $\bp$-closed
\[ \beta = \bp(||\Omega||^{-2}\, \gamma\lrcorner\bar\Omega)~.\]
Therefore
\[\beta = \bp\lambda \qquad\iff\qquad \alpha = \p^\dagger\gamma ~.\] 
Now, by assumption 
\[ H_{\p}^{(1,0)}(X) \cong H_{\bp}^{(0,1)}(X) = 0~,\]
which, when $\alpha$ is $\p^\dagger$-closed, it means by the Hodge decomposition of $\alpha$, that we must have that
in fact $\alpha = \p^\dagger\gamma$, and hence the corresponding element $\beta$ is $\bp$-exact.

Returning now to the moduli space, this result means that ${\cal Z}_t^{(0,2)}$ must be $\bp$-exact.
Hence, equations \eqref{eq:Dclosed} and \eqref{eq:eqsforZ} are equivalent and the {\it infinitesimal} moduli space is given by
equation \eqref{eq:defMSS}.  The condition \eqref{eq:novectors} is therefore sufficient to ensure that all deformations of the anomaly cancelation condition give rise to the holomorphic structure $\bar D$ on $\cal Q$ as required. 

This subtlety regarding deformations of the anomaly cancellation condition versus deformations of the holomorphic structure $\bar D$ deserves a bit more attention. First recall that $\bar D$ is a holomorphic structure on $\Q$ if and only if the Bianchi identities hold. Deformations of $\bar D$, which are the elements of $H^{(0,1)}_{\bar D}(\Q)$, therefore correspond to deformations of the Bianchi identities.  These correspond to deformations of the anomaly cancellation modulo $\d$-exact terms. One might think that in our scheme the deformations of the anomaly cancellation condition are only defined modulo $\d$-closed terms. However, due to flux quantisation, which states that the closed part of the flux, $H_0=\d B$, is quantised, we find that closed infinitesimal deformations of the anomaly cancellation condition must be exact\footnote{In~\citep{Anderson:2014xha} the authors discuss flux quantisation in relation to the deformations of the anomaly cancelation condition, but not in a slightly different context than ours.}.

It follows that the elements of $H^{(0,1)}_{\bar D}(\Q)$, i.e. deformations of the Strominger system, which of course includes the Bianchi identity, only define deformations of the anomaly cancellation modulo $\d$-exact terms. We can use this ambiguity to get rid of the $\p$-exact $(2,1)$-piece of the deformation of the anomaly cancellation condition. We might also get an extra $\bp$-exact piece, but this can be pulled into $\bp{\cal Z}_t^{(1,1)}$ by an appropriate redefinition of the $B$-field. In this way the $\p$-exact piece is trivial from the point of view of deformations of $\bar D$.

Finally, we note that equations \eqref{eq:eqsforZ} give a good interpretation of the elements in $H_{\bp}^1(X,T^*X)$ in the moduli space as the parameters for the complexified hermitian structure ${\cal Z}_t^{(1,1)}$ as defined in equations \eqref{eq:complexherm} and \eqref{eq:varB}\footnote{This is also obtained in~\citep{Candelas2014} from the dimensional reduction of the 10 dimensional heterotic string theory.}. These include the deformations of the $B$-field. It should also be noted that modding out the hermitian moduli by $\Im(\H_0)$ also makes sense from this perspective. Indeed, recall the gauge transformation of the $B$-field
\begin{equation}
\label{eq:gaugeB}
B_{t\,\textrm{gauge}}=-\frac{\a}{4}(\tr \d_A\epsilon_t-\tr \d_\Theta\eta_t)\:,
\end{equation}
required for the field-strength $H$ to remain invariant under gauge-transformations $\d_A\epsilon_t$ and $\d_\Theta\eta_t$ of $A$ and $\Theta$ respectively. It follows from this that trivial deformations $\tilde\B_t$, corresponding to gauge transformations, take the form
\begin{equation}
\label{eq:varBgauge}
\tilde\B_{t\,\textrm{gauge}}=-\frac{\a}{2}\big(\tr\,F\epsilon_t-\tr\,R\eta_t)+\frac{\a}{4}\d\big(\tr\,A\epsilon_t-\tr\,\Theta\eta_t\big)\:.
\end{equation}
At this point, we are not interested in gauge transformations that change the complex structures. That is, we set $\Delta_t=\bp\epsilon_t=\bp\eta_t=0$. It follows that $\eta_t=0$ by stability of $TX$. The same is true for $\epsilon_t$ if $V$ is stable. If $V=\oplus_iV_i$ is poly-stable, we may assume by \eqref{eq:gensec} that
\begin{equation*}
\epsilon_t=\sum_i c_i I_i\:,
\end{equation*}
and 
\begin{equation*}
\tilde\B_{t\,\textrm{gauge}}=\B^{(1,1)}_{t\,\textrm{gauge}}=-\frac{\a}{2}\sum_ic_i\tr\,F_i+\frac{\a}{4}\sum_ic_i\d\tr\,A_i=-\frac{\a}{4}\sum_ic_i\tr\,F_i\:,
\end{equation*}
where we have used that $\d\tr A_i=\tr F_i$ by symmetry of the trace. It follows that any term in $\Z_t^{(1,1)}$ which lies in $\Im(\H_0)$ should be considered trivial, and can thus be modded out.

Recall also that the $\bar D$-exact terms \eqref{eq:trivialy} included modding out by $\bp$-exact terms. For the $\p_t\omega^{(1,1)}$-term in ${\cal Z}_t^{(1,1)}$, this could be understood as preserving the conformally balanced condition. As for $\tilde\B_t^{(1,1)}$, recall that in addition to \eqref{eq:gaugeB}, the $B$-field also has the gauge transformation
\begin{equation}
\label{eq:gaugeB2}
B_{t\,\textrm{gauge}}=\d\lambda_t\:.
\end{equation}
The $(1,1)$-part of this reads
\begin{equation*}
B_{t\,\textrm{gauge}}^{(1,1)}=\p\lambda_t^{(0,1)}+\bp\lambda_t^{(1,0)}\:.
\end{equation*}
The first term is $\p$-exact, and can be understood as a trivial deformation of the {\it anti-holomorphic} $(1,0)$-type structure on $T^*X$. The last term corresponds to trivial deformations of the holomorphic striucture.

\newpage

\section{Conclusions and future directions.}
\label{sec:conclusions}

In this paper, we have discussed the first order deformations of the Strominger system. We have seen that the system can be described in terms of certain holomorphic structures on bundles over the base manifold $X$, and we have studied the first order deformations and moduli related to these structures. Studying first order deformations of holomorphic structures is easier than attacking heterotic compactificaitons head on, and first order deformations are given in terms of their corresponding first degree cohomologies. 

Indeed, the infinitesimal moduli space of the heterotic compactifications discussed in this paper is given by the tangent space of deformations of a bundle
\[{\cal Q} = T^*X\oplus \End(TX) \oplus \End(V) \oplus TX~,\]
endowed with a holomorphic structure defined by the operator $\bar D$ 
\begin{equation*}
\overline D =
\left(
\begin{array}{cc}
\bp & {\cal H}\\
0 & \bp_2\\
\end{array}
\right)~,\qquad
\bp_2 =
\left(
\begin{array}{ccc}
\bp_{\vartheta^I} & 0 & {\cal R}\\
0 & \bp_{\cal A} & {\cal F}\\
0 & 0 & \bp
\end{array}
\right)~.
\end{equation*}
This operator squares to zero due to the Bianchi identities for $F$, $R^I$ and $H$. We have shown that the  infinitesimal deformations of the Strominger system are given by the first order deformations of $\bar D$ which can be computed using the general theory of deformations of holomorphic bundles and we have found
\[ {\cal TM} = H_{\bar D}^{(0,1)}(X, \Q) \cong \Big[H_{\bp}^{(0,1)}(X, T^*X)\Big/\textrm{Im}({\cal H}_0)\Big] \oplus \ker{\cal H}_1~,\]
where 
\[\ker({\cal H}_1)\subseteq  H_{\bp_2}^{(0,1)}(X, E) ~,
\qquad E  = \End(TX) \oplus \End(V) \oplus TX~\]
and $\cal H$ is a map between the cohomologies
\begin{equation*}
{\cal H}\;:\;H_{\bp_2}^{(0,q)}(X, E)\longrightarrow H_{\bp}^{(0,q+1)}(X, T^*{}^{(1,0)}X)\:,
\end{equation*}
carefully constructed so that $\overline D^2=0$ is equivalent to the Bianchi identity for the anomaly cancelation condition.
Here $\bp_2$ is the holomorphic structure on $E$ given above. We also have
\begin{equation*}
\textrm{Im}({\cal H}_0)\cong\{\tr(F\alpha)\:\vert\:\alpha\in H^0_{\bp_{\cal A}}(X,\End(V))\},
\end{equation*}
which is trivial when $V$ is stable, as $H^0_{\bp_{\cal A}}(X,\End(V))=0$ in this case, but could be non-trivial if $V$ is polystable. The quotient by $\textrm{Im}({\cal H}_0)$ takes care of the constraints coming from the Yang-Mills condition.

A generic modulus therefore takes the form
\begin{equation*}
\left(
\begin{array}{c}
y \\ x\\
\end{array}
\right)
~,\qquad
y\in H_{\bp}^{(0,1)}(X, T^*X)\Big/\textrm{Im}({\cal H}_0)
~,\qquad
x=\left(
\begin{array}{c}
\kappa \\ \alpha\\ \Delta
\end{array}
\right)
\in \ker{\cal H} \subseteq H_{\bp_2}^{(0,q)}(X, E)~,
\end{equation*} 
where, 
\begin{equation*}
\Delta\in H^{(0,1)}_{\bp}(X, T^{(1,0)}X)\:,\;\;\;\;\alpha\in H^{(0,1)}_{\bp_A}(X,\End(V))\:,\;\;\;\;\kappa\in H^{(0,1)}_{\bp_I}(X,\End(TX))\:.
\end{equation*}
Here $\kappa$ appears as a generic element in $H^{(0,1)}_{\bp_I}(X,\End(TX))$ as a consequence of promoting an instanton connection $\nabla^I$ on $TX$ to a dynamical field. 
Finally, we argued that the factor $H_{\bp}^{(0,1)}(X, T^*X)$  of the moduli-space $\cal TM$ can be interpreted as complexified hermitian moduli, and these are given by the same cohomology as in the Calabi-Yau case.

\subsection{Discussion.}

As we have seen, promoting the connection $\nabla^I$ to a dynamical field gave us a first order moduli space $H^{(0,1)}_{\bp_I}(\End(TX))$ of deformations of $TX$ as a holomorphic bundle. 
These extra moduli are needed in order for the proposed mathematical structure to implement the anomaly cancelation condition. However, we do not believe they correspond  physical fields in the lower energy four-dimensional theory, nor do they appear in the heterotic string sigma-model.

It has been shown that a change of the connection in the ten-dimensional supergravity theory correspond to a field redefinition in the sigma-model perspective, at least at one-loop in the sigma-model \cite{Sen1986289}. In this way the connection $\nabla^I$ depends on the other fields of the theory, and this dependence comes down to how one defines the fields in the sigma-model.  However, not all field choices are {\it physical}, in the sense that they do not necessarily solve the equations of motion. This is what leads to the necessity of the instanton condition on the connection on $TX$.

For example, to first order in the $\alpha'$ expansion, a particular choice of fields leads to the connection $\nabla^-$. This connection is known as the Hull connection and is defined in the Appendix (it is the connection obtained from the Bismut connection by changing $H$ for $-H$). This connection does satisfy the instanton condition to the correct order in the first order theory. 

With this in mind it would be interesting to explore the possibility that the elements in $H^{(0,1)}_{\bp_I}(\End(TX))$ can be given an interpretation as the local moduli space of allowed field redefinitions for which the equations of motion are satisfied.
Indeed, when deforming the fields, the instanton connection $\nabla^I$ deforms correspondingly by an element $\kappa\in H^{(0,1)}_{\bp_I}(\End(TX))$. Depending on our field definition, there is an ambiguity in what $\kappa$ is, and perhaps this ambiguity is parameterised by $H^{(0,1)}_{\bp_I}(\End(TX))$. Note also that different field choices will deform $TX$ differently as a holomorphic bundle, where the new holomorphic structure is given by $\nabla^I+\kappa$ on a new bundle $(TX)'$. 
Of course, with this interpretation there would be nothing physical about these moduli, and they would not give rise to new fields in the lower energy theory. They would correspond to the fact that we have not specified what the sigma model field choice is. We have only specified that they are fields for which the equations of motion are also satisfied. We will discuss this more in a future publication~\citep{delaossa2014}.

Note that the allowed changes $\kappa$ of the instanton connection together with the actual physical moduli are further constrained by \eqref{eq:Dclosed}. This has interesting consequences in terms of moduli stabilisation. Indeed it has long been known that torsional compactifications give rise to further moduli stabilisation than in the K\"ahler case \cite{Dasgupta:1999ss, LopesCardoso:2003af, Becker:2006xp, Klaput:2012vv, Cicoli:2013rwa}. 

Finally, we note that it would be very interesting to compare the mathematical structure we constructed in this paper, that is the bundle $\Q$ over $X$ together its  holomorphic structure $\bar D$, with the work of~\citep{2013arXiv1304.4294G} and~\citep{2013arXiv1308.5159B} where torsion free generalised connections are applied to heterotic supergravity\footnote{See however~\citep{Anderson:2014xha}.}. In these papers, it seems they also need to promote the instanton connection on $TX$ to a field with its own equation of motion.

\subsection{Future directions.}
There a number of questions which would be interesting to pursue. We list a few of these here.

\subsubsection*{Metrics, obstructions and generalisations.}
It a very natural question to want to write a metric on moduli spaces.  In physics these correspond to kinetic terms in the effective four dimensional field theory. In our case, supersymmetry predicts that this metric should be K\"ahler.  In a forthcoming paper~\cite{Candelas2014} we compute this metric for certain compatifcations in heterotic string theory. 

It is a natural next step to try to work out the obstructions to the first order deformations of  the holomorphic structures introduced in section \ref{sec:moduli}, to discover which first order deformations survive to higher orders. 
Moreover, it would be interesting to generalise the study to the case where the torsion class $W_1^{\omega}$ is not exact.  
This is interesting mathematically as this corresponds to the study of complex  manifolds $X$ with canonical bundle which is only {\it locally} conformally trivial.  Physically this would mean that the compactification is not supersymmetric anymore (recall supersymmetry requires that $W_1^{\omega}=\d\phi$) and one would have to study the equations of motion to try to figure out whether they in fact correspond to allowed compactifications.

\subsubsection*{Preservation of the heterotic structure and embeddings into G-structures.}
It would also be interesting to study the heterotic $SU(3)$-structure more, as introduced in section \ref{subsec:Xgeom}. This structure is interesting in its own right, both from a physics and mathematics perspective. Indeed it would be interesting to see what complex structure deformations survive in this structure to higher orders. As we have seen, a condition on the first order deformations $\chi$ of the holomorphic form $\Omega$ is that
\begin{equation*}
\chi\in H^{(2,1)}_{\d}(X)\:.
\end{equation*}
In~\citep{DKS2014} we elaborate on these matters.  Also, as mentioned in Section \ref{subsec:infsu3}, in \citep{DKS2014} we also study the heterotic structure by embedding it into another $G$-structure, like $G_2$-, $SU(4)$-, or $\text{Spin}(7)$-structures. We show that requiring families of manifolds with a heterotic  $SU(3)$ structure to have certain $G$-structures guarantees that the heterotic structure, and in particular the difficult conformally balanced condition, is preserved along the family. 
We find that this may have very interesting for applications to $F$-theory and $M$-theory.

\subsubsection*{$\alpha'$-corrections and relations to physics.}
In an upcoming publication \cite{delaossa2014} we consider heterotic supergravity at second and higher orders in $\a$. We will see that supersymmetric solutions of Strominger type survive to second order in $\a$. Moreover they appear to be generic. We comment on the connection choice on the tangent bundle which we again find should satisfy the instanton condition. In particular, we note that the choice of connection is as though the connection was a dynamical field, with its own supersymmetry condition. We make conjectures on what the connection and the geometry should be at higher orders in $\a$. 

It would interesting too to undertand better the physics behind the moduli space derived in this paper. In particular, it would be interesting to
generalise the analysis in \cite{Anderson:2010mh}, where a superpotential is generated for moduli not in the kernel of the map $\F$. We would expect that similar superpotential terms appear in the non-K\"ahler case. In particular, we would expect a  superpotential~\citep{LopesCardoso:2003af, Becker:2003yv, Becker:2003gq}
\begin{equation*}
W=\int_X(H+i\d\omega)\wedge\Omega\:.
\end{equation*}
to be generated in the four-dimensional theory whenever the map $\H$ is non-trivial.

\section{Acknowledgements}
We would like to thank Philip Candelas, Andrei Constantin, Andrew Dancer, Spiro Karigiannis, Magdalena Larfors,  Zhentao Lu, Jock McOrist, James Gray, and Ruxandra Moraru.  Also we would like to thank L.~Anderson, J.~Gray and E.~Sharpe for agreeing to coordinate with us the submission to the ArXiv of their paper~\citep{Anderson:2014xha} which overlaps with the work  presented in this paper.
XD would like to thank Perimeter Institute and ICTP in Trieste for hospitality while some of this work was carried through.  
XD's research is supported in part by the EPSRC grant BKRWDM00. 
ES is supported by the Clarendon Scholarship of OUP, and a Balliol College Dervorguilla Scholarship. 

\newpage
\appendix

\section{First Order Heterotic Supergravity.}
\label{sec:FirstOrder}
In this appendix we review heterotic supergravity at first order in $\a$. We write down the action and supersymmetry transformations, and review the supersymmetric solutions of this theory, commonly known as the Strominger Sytem \cite{Strominger:1986uh, Hull:1986kz}. We describe how consistency between the supersymmetry conditions and equations of motion constrains the choice of connection in the action. Various proofs of this have appeared in the literature before \cite{Hull:1986kz, Bergshoeff:1988nn, Ivanov:2009rh, Martelli:2010jx}, and we give a slightly different proof in this appendix. We also comment on the type of geometry that results from the first order supersymmetry conditions, and in particular the fact the compact space $X$ is conformally balanced.

\subsection{Action and Field Content.}
Let's begin by recalling the action at this order \cite{Bergshoeff1989439}
\begin{equation}
\label{eq:action}
S=\frac{1}{2\kappa_{10}^2}\int_{M_{10}}e^{-2\phi}\Big[*\mathcal{R}-4\vert\d\phi\vert^2+\frac{1}{2}\vert H\vert^2+\frac{\alpha'}{4}(\tr\vert F\vert^2-\tr\vert R\vert^2)\Big]+\OO(\a^2).
\end{equation}
$F$ is now the curvature of the $E_8{\times}E_8$ gauge bundle, $R$ is the curvature of the tangent bundle, while the NS-NS three-form,
\begin{equation}
\label{eq:anomalycancellation}
H=\d B+\frac{\a}{4}(CS[A]- CS[\Theta]),
\end{equation}
is appropriately defined for the theory to be anomaly free. Here, the $CS[A]$ and $CS[\Theta]$ are Chern-Simons three-forms of the gauge-connection $A$, and the tangent bundle connection $\Theta$, respectively. The choice of connection $\Theta$ has been a subtle issue which at times has been confusing in the literature. It has been argued that changing the connection is equivalent to a field redefinition \cite{Sen1986289}. This does not however give us the freedom to choose whatever connection we prefer, as we also need a connection choice for which a solution of the supersymmetry equations together with the anomaly cancelation condition is a solution of the equations of motion. 
We return to this point later in this Appendix.

At first order in $\a$, the  Hull connection $\nabla^-$ whose connection symbols are 
\begin{equation}
\label{eq:Hull}
{\Gamma^-_{KL}}^M={\Gamma^{LC}_{KL}}^{\: M}-\frac{1}{2}{H_{KL}}^M
\end{equation}
does tick this box. That this connection gives compatibility between supersymmetry and the equations of motion, was first noted by Hull \cite{Hull:1986kz}. Furthermore, it turns out that this is the connection choice that leaves the full action invariant under supersymmetry transformations at first order in $\a$ \citep{Bergshoeff1989439}.  We will discuss this further in~\citep{delaossa2014}.

The fermonic content of the theory is the gravitino $\psi$, the dilatino $\lambda$ and the gaugino $\chi$. The supersymmetry transformations usually take the form~\citep{Bergshoeff1989439}
\begin{align}
\label{eq:O1spinorsusy1}
\delta\psi_M &=\nabla^+_M\epsilon=\Big(\nabla_M+\frac{1}{8}\H_M\Big)\epsilon+\OO(\a^2) \\
\label{eq:O1spinorsusy2}
\delta\lambda&=\Big(\slashed\nabla\phi+\frac{1}{12}\H\Big)\epsilon+\OO(\a^2) \\
\label{eq:O1spinorsusy3}
\delta\chi&=F_{MN}\Gamma^{MN}\epsilon+\OO(\a^2) .
\end{align}
where $\epsilon$ is a ten-dimensional Majorana-Weyl spinor parameterising supersymmetry, $\Gamma^M$ are ten-dimensional gamma-matrices, $\H_M=H_{MNP}\Gamma^{NP}$, $\H=H_{MNP}\Gamma^{MNP}$. We use large roman indices to denote indices on $M_{10}$. We have also defined the connection $\nabla^+$ by it's connection symbols
\begin{equation*}
{\Gamma^+_{KL}}^M={\Gamma^{LC}_{KL}}^M+\frac{1}{2}{H_{KL}}^M\:.
\end{equation*}
Moreover, under the supersymmetry transformations \eqref{eq:O1spinorsusy1}-\eqref{eq:O1spinorsusy3}, we need to use the Hull connection \eqref{eq:Hull} in the action, in order to have a supersymmetry invariant theory. Supersymmetry requires that we set \eqref{eq:O1spinorsusy1}-\eqref{eq:O1spinorsusy3} to zero. Upon compactification to four dimensions, this leads to supersymmetric solutions of the Strominger system, described in section \ref{sec:hetcomp}, and which we briefly review next.

Before we do so, we note that the connection $\nabla^-$ in the action can be changed. The price we pay in doing so is that the supersymmetry transformations \eqref{eq:O1spinorsusy1}-\eqref{eq:O1spinorsusy3} should change as well. However, if we still insist that the supersymmetric solutions constrain the geometry of $X_6$ such that it still satisfies the Strominger system \footnote{In \cite{delaossa2014} we show that this can be done without loss of generality.}, then this imposes conditions on what the changes to the connection can be. In particular, it forces the connection to remain an $SU(3)$-instanton as we show below. This is all explained in greater detail in \cite{delaossa2014}.

\subsection{First Order Supersymmetry and Geometry.}
As in section \ref{sec:hetcomp}, the ten dimensional manifold is taken to be the product,
\begin{equation*}
M_{10}=M_4\times X_6,
\end{equation*}
where $M_4$ is four-dimensional space-time, and $X_6$ is a compact internal space. Let's take a moment to recall what conditions supersymmetry imposes from the set of transformations \eqref{eq:O1spinorsusy1}--\eqref{eq:O1spinorsusy2} on the internal geometry of $X_6$ (see summary in section \ref{subsec:sum}), commonly known as the Strominger system.  Introducing the fields $(\Psi, \omega)$ as in section \ref{subsec:Xgeom} these constraints may be written as

\begin{align}
\label{eq:susy1}
\d(e^{-2\phi}\Psi)&=0\\
\label{eq:susy2}
\d(e^{-2\phi}\omega\wedge\omega)&=0\\
\label{eq:susy3}
-e^{2\phi}\d(e^{2\phi}\omega)&=*H,
\end{align}
From \eqref{eq:anomalycancellation} we also get the following Bianchi identity
\begin{equation}
\d H=\frac{\a}{4}(\tr (F\wedge F)-\tr (R\wedge R))\:.
\end{equation}
Setting the gaugino variation \eqref{eq:O1spinorsusy3} to zero is also equivalent to requiring that the gauge-bundle is holomorphic and satisfies the hermitian Yang-Mills equations on the internal space
\begin{equation}
\label{eq:inst1}
F\wedge\Omega=0,\;\;\;\; F\wedge\omega\wedge\omega=0,
\end{equation}
where $F$ is the field-strength of the $E_8{\times}E_8$ gauge-bundle.

Equation \eqref{eq:susy1} implies the existence of a holomorphic three-form $\Omega=e^{-2\phi}\Psi$, and it also implies that the complex structure $J$ determined by $\Psi$ is integrable.
A complex three-fold $X_6$ satisfying equation \eqref{eq:susy2} is said to be conformally balanced. In this case, the Lee-form $W_1^\omega$ of $\omega$ is identified with $\d\phi$
\begin{equation}
W_1^\omega=\frac{1}{2}\omega\lrcorner\d\omega=\d\phi.
\end{equation}
By taking the Hodge dual of equation \eqref{eq:susy3} we find~\citep{Strominger:1986uh} 
\begin{equation}
\label{eq:flux2}
H=i(\p-\bar\p)\omega\:.
\end{equation}

Having discussed supersymmetry at first order, we also need to discuss the equations of motion. That is, satisfying supersymmetry and the Bianchi Identity does not guarantee a solution to the equations of motion. However, at first order in $\a$ it does, provided one chooses the correct connection for the curvature two-form appearing in the action and in the Bianchi Identity. We will see how this works next.

\subsection{Instanton Condition.}
It has been shown that the supersymmetry equations and the Bianchi identity, \eqref{eq:susy1}-\eqref{eq:inst1}, imply the equations of motion if and only if the connection $\nabla$ for the curvature two-form $R$ appearing in \eqref{eq:action} is an $SU(3)$-instanton \cite{Ivanov:2009rh, Martelli:2010jx}. This means that it satisfies the conditions
\begin{equation}
\label{eq:inst2}
R\wedge\Omega=0,\;\;\;\;R\wedge\omega\wedge\omega=0,
\end{equation}
similar to the field-strength $F$. 

We now give our own proof of the instanton condition \eqref{eq:inst2}. In \cite{LopesCardoso:2003af} it was shown that the six-dimensional part of the action \eqref{eq:action} may then be rewritten in terms of the $SU(3)$-structure forms, using the Bianchi identity, as 
\begin{align}
\label{eq:6daction}
S_6&=\frac{1}{2}\int_{X_6}e^{-2\phi}\Big[-4\vert\d\phi-W_1^\omega\vert^2+\omega\wedge\omega\wedge\mathcal{\hat R}+\vert H-e^{2\phi}*\d(e^{-2\phi}\omega)\vert^2\Big]\notag\\
&-\frac{1}{4}\int\d ^6y\sqrt{g_6}{N_{mn}}^pg^{mq}g^{nr}g_{ps}{N_{nq}}^s\notag \\
&-\frac{\alpha'}{2}\int\d ^6y\sqrt{g_6}e^{-2\phi}\Big[\tr\vert F^{(2,0)}\vert^2+\tr\vert F^{(0,2)}\vert^2+\frac{1}{4}\tr\vert F_{mn}\omega^{mn}\vert^2\Big]\notag\\
&+\frac{\alpha'}{2}\int\d ^6y\sqrt{g_6}e^{-2\phi}\Big[(\tr\vert{R}^{(2,0)}\vert^2+\tr\vert {R}^{(0,2)}\vert^2+\frac{1}{4}\tr\vert R_{mn}\omega^{mn}\vert^2\Big]+\OO(\a^2).
\end{align}
Here $\mathcal{\hat R}$ is the Ricci-form of the unique almost complex structure compatible metric connection $\hat\nabla$ that also has totally antisymmetric torsion (recall that  the almost complex structure is determined by $\Psi$, see equations \eqref{eq:complexstr} and \eqref{eq:CSnorm}). This connection is known as the Bismut connection in the mathematics literature. The Ricci-form is given by
\begin{equation}
\mathcal{\hat R}=\frac{1}{4}\hat R_{pqmn}\omega^{mn}\d x^p\wedge\d x^q,
\end{equation}
while ${N_{mn}}^p$ is the Nijenhaus tensor for this almost complex structure. 

After variations of the action at the supersymmetric locus, we find that most of the terms vanish to the given order. Indeed,  we saw above that first order supersymmetry implies that $W_1^\omega=\d\phi$.  From equation \eqref{eq:susy3}, we get that the the term involving $H$ vanishes. The Nijenhaus tensor ${N_{mn}}^p$ vanishes, as $J$ is integrable, while the terms involving the bundle vanish due to \eqref{eq:inst1}. Finally, \eqref{eq:flux2} identifies $\nabla^+$ with $\hat\nabla$, which implies that $\hat\nabla$ has $SU(3)$-holonomy, and hence  the vanishing of $\mathcal{\hat R}$. We are thus left with
\begin{align}
\label{eq:6dvary}
\delta S_6&=\frac{1}{2}\int_{X_6}e^{-2\phi}\omega\wedge\omega\wedge\delta\mathcal{\hat R}\notag\\
&+\frac{\alpha'}{2}\delta\int\d ^6y\sqrt{g_6}e^{-2\phi}\Big[(\tr\vert{R}^{(2,0)}\vert^2+\tr\vert {R}^{(0,2)}\vert^2+\frac{1}{4}\tr\vert R_{mn}\omega^{mn}\vert^2\Big]+\OO(\a^2).
\end{align}
In \cite{LopesCardoso:2003af} it was also shown that $\delta\mathcal{\hat R}$ is exact. The conformally balanced condition \eqref{eq:susy2}, implies that the first term of \eqref{eq:6dvary} vanishes. We thus see that $R$ needs to satisfy \eqref{eq:inst2} in order for the action to be extremized at the supersymmetric locus. In fact, to the order in $\alpha'$ we are working at, we find that we need
\begin{equation}
R_{mn}\Gamma^{mn}\eta= 0 + \OO(\a),
\end{equation}
in order  for $\delta S_6=0 + \OO(\a^2)$. This reduction in $\a$-orders comes from the extra factor of $\a$ in curvature terms above.

\subsection{The Hull Connection.}
We show in this subsection that to first order in $\alpha'$, it is sufficient to use the Hull connection in the action, as this satisfies the instanton condition to the order we need. This argument has appeared in the literature before~\citep{Martelli:2010jx}, but we include it here for completeness.

By a direct computation, one finds the identity
\begin{equation}
\label{eq:curvatureid}
R^+_{mnpq}-R^-_{pqmn}=\frac{1}{2}\d H_{mnpq}.
\end{equation}
This implies that 
\begin{equation}
R^+_{mnpq}=R^-_{pqmn}+\OO(\a),\label{eq:PlusMinus}
\end{equation}
by the Bianchi identity. Contracting both sides with $\Gamma^{pq}\eta$ and using
\begin{equation}
R^+_{mnpq}\Gamma^{pq}\eta=0+\OO(\a^2),
\end{equation}
which is the integrability condition for the supersymmetry equation \eqref{eq:O1spinorsusy1}, stating that $\nabla^+$ has $SU(3)$-holonomy,  equation \eqref{eq:PlusMinus} gives
\begin{equation}
R^-_{pqmn}\Gamma^{pq}\eta=0+\OO(\a)\:.
\end{equation}
Hence, we see that to this order in the $\alpha'$ expansion, the action \eqref{eq:6daction} is extremised when supersymmetry and the Bianchi identity are satisfied, if we take $\nabla$ to be the Hull connection $\nabla^-$. 
Note that by choosing the Hull connection the instanton condition is always satisfied at the given order. Hence it does not put any further restrictions on the first order theory. At higher orders in $\alpha'$, this is no longer the case. 

\subsection{Note on the Choice of Connection and Higher Order Corrections.}
As mentioned before, it has been shown that the choice of connection corresponds to a choice of field definitions in the sigma-model point of view \cite{Sen1986289}. One might therefore think that any connection choice should be valid, as it just corresponds to a field redefinition. This turns out to be wrong. Indeed, if we insist that {\it supersymmetric solutions satisfy the Strominger system}, we need to choose the fields so that the equations of motion are compatible. As we have discussed, this turns out to restrict the connection to satisfy the instanton condition at $\OO(\a)$. 

This condition on the connection might receive corrections at higher orders in $\a$, and we now discuss these corrections and the condition on the connection at higher orders in $\a$. We comment briefly on the condition on the connection at $\OO(\a^2)$ as we describe this at length in a forthcomming publication \citep{delaossa2014}.

We note that there is a symmetry between the tangent bundle connection $\nabla$ and gauge connection $A$ in the action \eqref{eq:action}.  As a guiding principle, as is also done in \cite{Bergshoeff1989439}, we would like to keep this symmetry to higher orders. {\it In particular, we suggest that $\nabla$ should be chosen to satisfy an equation of motion similar to that of $A$}. Indeed, this is precisely what one would get  if $\nabla$ was dynamical.
We claim that 
\begin{itemize}
\label{tm:inst}
\item{\it
Strominger type supersymmetric solutions, where $\nabla^+\epsilon=0$ for heterotic compactifications on a compact six-fold $X$, survive as solutions of heterotic supergravity at $\OO(\a^2)$ if and only if the connection $\nabla$ satisfies the instanton condition,
\begin{equation*}
R_{mn}\Gamma^{mn}\eta=0\:.
\end{equation*}
For compact spaces these solutions appear to be generic. Moreover, $\nabla$ satisfies it's own equation of motion for these solutions.}
\end{itemize}
It also turns out that with such a connection choice, the first order equations of motion remain the same also to second order in $\a$. This makes sense from a world-sheet point of view. Indeed, in \cite{Foakes1988335} it was noted that the two-loop $\beta$-function of the gauge connection equals the one-loop $\beta$-function. That is, the $\beta$-function of the gauge field does not receive corrections at this order, so nor should the equation of motion. This is also what one finds from the supergravity point of view~\citep{Bergshoeff1989439}. It is therefore perhaps no surprise that when we choose $\nabla$ so that it satisfies it's own equation of motion\footnote{Note that this equation is adequately satisfied by the Hull connection in the first order theory.}
\begin{equation*}
e^{2\phi}\d_\nabla(e^{-2\phi}*R)-R\wedge*H=0~,
\end{equation*}
the other equations of motion remain the same. From the world-sheet point of view, the equations of motion should correspond to the sigma-model $\beta$-functions. The particular field choice above is thus such as the $\beta$-functions don't receive corrections at $\OO(\a^2)$. Due to the symmetry between the gauge connection $A$ and $\Theta$ in the action, we expect that such a field choice is possible. Indeed, as mentioned above the equation of motion of the gauge connection $A$ remains the same at $\OO(\a^2)$`\citep{Bergshoeff1989439},
\begin{equation*}
e^{2\phi}\d_A(e^{-2\phi}*F)-F\wedge*H=0\:,
\end{equation*}
and so does it's two-loop beta-function. We therefore expect a similar story for $\nabla$. Indeed we find that compatibility between space-time supersymmetry and the equations of motion seems to require the existence of such a connection.

\newpage

\bibliographystyle{JHEP}
\bibliography{BibliographyXD}

\end{document}